\documentclass[11pt]{article}
\usepackage{setspace}
\usepackage[utf8]{inputenc}
\usepackage[english]{babel}
\usepackage{color,amstext, amsmath,latexsym,amsbsy,amssymb,amsthm,graphicx}
\usepackage{mathrsfs}
\usepackage{txfonts}
\usepackage{amsfonts}
\usepackage{exscale,relsize,makeidx}
\newtheorem{Lemma}{Lemma}[section]
\newtheorem{Definition}[Lemma]{Definition}
\newtheorem{Remark}[Lemma]{Remark}
\newtheorem{Proposition}[Lemma]{Proposition}
\newtheorem{Theorem}[Lemma]{Theorem} 
\newtheorem{Corollary}[Lemma]{Corollary}

\newcounter{mnotecount}[section]

\begin{document}
\title{Applications of Fixed Point Theorems to the Vacuum Einstein Constraint Equations with Non-Constant Mean Curvature}
\author{The Cang Nguyen}
\date{May 28, 2014}
\maketitle
\begin{abstract}
In this paper, we introduce new methods for solving the vacuum Einstein constraints equations: the first one is based on Schaefer's fixed point
theorem (known methods use Schauder's fixed point theorem) while the second one uses the concept of half-continuity coupled with the introduction
of local supersolutions. These methods allow to: unify some recent existence results, simplify many proofs (for instance, the one of the main theorem
in \cite{DGH}) and weaken the assumptions of many recent results.
\end{abstract}

\tableofcontents

\section{Introduction}
The Einstein equations for a $(n+1)-$manifold $\mathscr{M}$ ($n\geq 3$) and a Lorentzian metric $h$ describe the evolution of the gravitational field. In the vacuum case, they read
\begin{equation}\label{Einstein.eq}
\text{Ric}_{\mu\upsilon}-\frac{R}{2}h_{\mu\upsilon}=0.
\end{equation}
Here $\text{Ric}$ and $\text{R}$ are respectively the Ricci and the scalar curvature of $\text{h}$. If $M$ is a two sided spacelike hypersurface with unit normal $\nu$ of this geometric
space, one can define on $M$
\begin{itemize}
\item the induced metric $g=h_{|TM}$,
\item the second fundamental form $K$: $K(X, Y) = h({}^h\nabla_X \nu, Y)$.
\end{itemize}
It follows from the Gauss and Codazzi equations that $g$ and $K$ are related by the so-called constraint equations
\begin{equation}\label{Einstein.eq2}
\begin{aligned}
R_{g}-|K|_{g}^{2}+\left(\text{tr}_{g}K\right)^{2}&=0,\\
\text{div}_{g}K-d~\text{tr}_{g}K&=0.
\end{aligned}
\end{equation}
Y. Choquet-Bruhat and R. Geroch proved in \cite{CB,CBG69} that if a $n-$manifold $M$, a Riemannian metric $g$ on $M$ and a $2-$tensor $K$ form a solution to Equations \eqref{Einstein.eq2}, we can reconstruct
the space-time $(\mathscr{M},h)$ from the knowledge of $(M,g,K)$.\\

In an effort to solve \eqref{Einstein.eq2}, A. Lichnerowicz \cite{Lichnz} and later Y. Choquet-Bruhat and J. York \cite{CBY80} introduced a very efficient approach called the conformal method. For this procedure,
let $M$ be a $n-$manifold and $g$ be a Riemannian metric on $M$. One specifies a mean curvature $\tau$ and a transverse-traceless tensor $\sigma$ (i.e. a symmetric, trace-free, divergence-free $(0,2)-$tensor)
on $M$. One looks for a conformal factor $\varphi~:~M\longrightarrow\mathbb{R}$ and a $1-$form $W$ such that
\begin{equation}
 \left\{\begin{array}{ll}
 \widetilde{g} &=\varphi^{N-2}g\\
 \widetilde{K}&=\frac{1}{n}\tau\widetilde{g}+\varphi^{-2}(\sigma+LW)
 \end{array} \right.
 \end{equation}
form a solution to the Einstein equations \eqref{Einstein.eq2}. Here $N=2n/(n-2)$ and $L$ is the conformal Killing operator defined by
$$LW_{ij}=\nabla_{i} W_{j}+\nabla_{j}W_{i}-\frac{2}{n}(\mathrm{div}W)g_{ij},$$
where $\nabla$ denotes the Levi-Civita connection associated to the metric $g$ and $\mathrm{div} W = \nabla^i W_i$ is the divergence operator.
Equations \eqref{Einstein.eq2} are then reformulated into the following coupled nonlinear elliptic system for a positive $\varphi$ and a vector field $W$:
\begin{subequations}\label{CE}
\begin{eqnarray}
\frac{4(n-1)}{n-2}\Delta \varphi+R\varphi&=&-\frac{n-1}{n}\tau^{2}\varphi^{N-1}+|\sigma+LW|^{2}\varphi^{-N-1}~~{\footnotesize [\text{Lichnerowicz equation}]}\\
-\frac{1}{2}L^{*}LW&=&\frac{n-1}{n}\varphi^{N}d\tau ~~~~~~~~~~~~~~~~~~~~~~~~~~~~~~~\qquad{\footnotesize [\text{vector equation}]},
\end{eqnarray}
\end{subequations}
where $\Delta$ is the nonnegative Laplace operator, $R$ is the scalar curvature of $g$ and $L^{*}$is the formal $L^{2}-$adjoint of $L$, so $-\frac{1}{2} (L^*L W)_j = \nabla^i (LW_{ij})$.
These coupled equations are called the \textit{conformal constraint equations}. When $\tau$ is constant, the system \eqref{CE} becomes uncoupled (since $d\tau=0$ in the vector equation) and a
complete description of the situation was achieved by Isenberg (see \cite{BI}). When $\tau$ is not constant, the problem is much harder and there are still many situations where the solvability
of the system \eqref{CE} is not known. Recently, much progress has been made by several authors. Let us cite for instance:
\begin{itemize}
\item Isenberg-Moncrief \cite{IM}, Maxwell \cite{Maxwell05} for near CMC-results (i.e. $\tau$ close to constant),
\item Holst-Nagy-Tsogtgerel \cite{HNT} and Maxwell \cite{Maxwell} for far from CMC-results with a smallness assumption on $\sigma$, depending only on $g$ and $\tau$.
\item Dahl-Gicquaud-Humbert \cite{DGH} who proved that non-existence of solutions to a certain limit equations ensures the solvability of \eqref{CE}.
\end{itemize}
In this paper, we develop two new methods for solving the coupled system \eqref{CE}. The first one is based on Schaefer's fixed points which turns out to be more efficient in this situation than an application
of Schauder's fixed point theorem as used in \cite{DGH}, \cite{HNT} and \cite{Maxwell}. This method has several applications. In particular, it greatly simplifies the proof of the main theorem in \cite{DGH}
(see Theorem \ref{theoremofDHG}) and allows to recover an existence result provided $\sigma$ is small enough in $L^\infty$ (depending only on $g$ and $\tau$) as noticed in \cite{HNT} and \cite{Maxwell} (see
Proposition \ref{farCMCofMawell}). Furthermore, it gives an unifying point of view of these results. It is also worth noting that another effort to obtain the far from CMC-result has been recently presented
in \cite{G-NQA} using the implicit function theorem.\\

The second method uses half-continuity of appropriate maps. It allows to show that the assumption of the existence of global supersolutions used in \cite{DGH}, \cite{HNT} and \cite{Maxwell} to solve
\eqref{CE} can be weakened: the existence of local supersolutions, whose definition is given in Section 4, is sufficient here. As applications of this method, we prove the solvability of a modification of
the system \eqref{CE} when $\tau$ has some zeros and we show that the smallness of $\sigma$ in $L^{2}$ leads to the solvability of \eqref{CE}. This improves the results by Holst-Nagy-Tsogtgerel
\cite{HNT} and Maxwell \cite{Maxwell}.\\

In Section 2, we introduce the notations which will appear in the whole paper and we establish some  general results used in many proofs. In Section 3, we show how Schaefer's fixed point theorem
can be used to solve \eqref{CE}. We apply it to give a simpler proof of the main result in \cite{DGH} (see Theorem \ref{theoremofDHG}) and enlighten several consequences of this method. In Section 4, we
introduce the half-continuity method and give some applications.\\

\section*{Acknowledgements}
\addcontentsline{toc}{section}{Acknowledgements}
The author would like to thank Romain Gicquaud and Emmanuel Humbert for their advice, helpful discussions and great patience in their careful reading of preliminary versions of this article. 

\section{Preliminaries}
Let  $M$ be a compact manifold of dimension $n\geq 3$, our goal is to find solutions to the vacuum Einstein equations using the conformal method. The given data on $M$ consist in
\begin{equation}\label{condintial}
\begin{aligned}
&\bullet~~\mbox{a Riemannian metric $g\in W^{2,p},$}\\
&\bullet~~\mbox{a function $\tau\in W^{1,p}$,}~~~~~~~~~~~~~~~~~~~~~~~~~~~~~~~~~~~~~~~~~~~~~~~~~~~~~~~~~~~~~~~~~~~~~~~~~~~~~~~~~~~~~~~~~~~~~~~~~~~~~~~~~~\\
&\bullet~~\mbox{a symmetric, trace- and divergence-free $(0,2)-$tensor $\sigma\in W^{1,p}$,}
\end{aligned}
\end{equation}
with $p>n$. And one is required to find
\begin{equation*}
\begin{aligned}
&\bullet~~\mbox{a positive function $\varphi\in W^{2,p}$,}~~~~~~~~~~~~~~~~~~~~~~~~~~~~~~~~~~~~~~~~~~~~~~~~~~~~~~~~~~~~~~~~~~~~~~~~~~~~~~~~~~~~~~~~~~~~~~\\
&\bullet~~\mbox{a $1-$form $W\in W^{2,p}$,}
\end{aligned}
\end{equation*}
which satisfy the conformal constraint equations \eqref{CE}. We also assume that
\begin{equation}\label{condinitial2}
\begin{aligned}
&\bullet~~\mbox{$Z(\tau)$ has zero Lebesgue measure,}\\
&\bullet~~\mbox{$(M,g)$ has no conformal Killing vector field,}~~~~~~~~~~~~~~~~~~~~~~~~~~~~~~~~~~~~~~~~~~~~~~~~~~~~~~~~~~~~~~~~~~\\
&\bullet~~\mbox{ $\sigma\nequiv 0$ if $\mathcal{Y}_{g}\geq 0$,}
\end{aligned}
\end{equation}
where $Z(\tau)=\tau^{-1}(0)$ denotes the set of zero points of $\tau$ and $\mathcal{Y}_{g}$ is the Yamabe constant of the conformal class of $g$; that is
$$\mathcal{Y}_{g}=\inf_{\substack{f\in C^{\infty}(M)\\f\nequiv 0}}\frac{\frac{4(n-1)}{n-2}\int_{M}{|\nabla f|^{2}dv}+\int_{M}{Rf^{2}}}{\|f\|^{2}_{L^{N}(M)}}.$$
 We use standard notations for function spaces, such as $L^{p}$, $C^{k}$, and Sobolev spaces $W^{k,p}$. It will be clear from the context if the notation refers to a space of functions on $M$, or a space of
 sections of some bundle over $M$. For spaces of functions which embed into $L^{\infty}$, the subscript $+$ is used to indicate the cone of positive functions.\\

 We will sometimes write, for instance, $C(\alpha_{1},\alpha_{2})$ to indicate that a constant $C$ depends only on $\alpha_{1}$ and $\alpha_{2}$.\\
 
 From now on,  we define the map $T:~L^{\infty}\rightarrow L^{\infty}$ as follows. Given data on $M$ as specified in \eqref{condintial} and assuming that \eqref{condinitial2} holds,
 for each $\varphi\in L^{\infty}$, there exists a unique $W\in W^{2,p}$ such that
 $$-\frac{1}{2}L^{*}LW=\frac{n-1}{n}\varphi^{N}d\tau,$$
 and there is a unique $\psi\in W^{2,p}_{+}$ satisfying (see \cite{IM} or \cite{Maxwell05})
 $$\frac{4(n-1)}{n-2}\Delta \psi+R\psi=-\frac{n-1}{n}\tau^{2}\psi^{N-1}+|\sigma+LW|^{2}\psi^{-N-1}.$$
 We define $$T(\varphi)=\psi.$$ 
\begin{Proposition}\label{compacityofT} (see \cite[Lemma 2.3]{DGH} or \cite{Maxwell})
T is a continuous compact operator and  $T(\varphi)>0$ for all $\varphi\in L^{\infty}$.
\end{Proposition}
We now review some standard facts on the Lichnerowicz equation on a compact $n-$manifold $M$:
\begin{equation}\label{Lichnerowicz}
\frac{4(n-1)}{n-2}\Delta u+Ru+\frac{n-1}{n}\tau^{2}u^{N-1}=\frac{w^{2}}{u^{N+1}}.
\end{equation}
Given a function $w$ and $p>n$, we say that $u_{+}\in W_{+}^{2,p}$ is a \textit{supersolution} to \eqref{Lichnerowicz} if
$$\frac{4(n-1)}{n-2}\Delta u_{+}+Ru_{+}+\frac{n-1}{n}\tau^{2}u_{+}^{N-1}\geq\frac{w^{2}}{u_{+}^{N+1}}.$$
A \textit{subsolution} is defined similarly with the reverse inequality.
\begin{Proposition}[see \cite{Maxwell05}] \label{method sub-super}
Assume that $g\in W^{2,p}$ and $w,\tau\in L^{2p}$ for some $p>n$. If $u_{-},u_{+}\in W_{+}^{2,p}$ are a subsolution and a supersolution respectively to \eqref{Lichnerowicz} associated with a fixed $w$
such that $u_{-}\leq u_{+}$, then there exists a solution $u\in W_{+}^{2,p}$ to \eqref{Lichnerowicz} such that $u_{-}\leq u\leq u_{+}$.
\end{Proposition}
\begin{Theorem}[see \cite{Maxwell05}]\label{maxwell1}
Assume $w,\tau\in L^{2p}$ and $g\in W^{2,p}$ for some $p>n$. Then there exists a positive solution $u\in W^{2,p}_{+}$ to \eqref{Lichnerowicz} if and only if one of the following assertions is true.
\begin{enumerate}
\item $\mathcal{Y}_{g}>0$ and $w\nequiv 0$,
\item $\mathcal{Y}_{g}=0$ and $w\nequiv 0$, $\tau\nequiv 0$,
\item $\mathcal{Y}_{g}<0$ and there exists $\hat{g}$ in the conformal class of $g$ such that $R_{\hat{g}}=-\frac{n-1}{n}\tau^{2}$,
\item $\mathcal{Y}_{g}=0$ and $w\equiv 0$, $\tau\equiv 0.$
\end{enumerate}
In Cases $1-3$ the solution is unique. In Case $4$ any two solutions are related by a scaling by a positive constant multiple. Moreover, Case $3$ holds if $\mathcal{Y}_{g}<0$ and $Z(\tau)$
has zero Lebesgue measure (see \cite[Theorem 6.12]{Au}).
\end{Theorem} 
The main technique used to prove the above theorem is the conformal covariance of \eqref{Lichnerowicz}.
\begin{Lemma}[see \cite{Maxwell}]\label{maxwell2}
Assume that $g\in W^{2,p}$ and that $w, \tau\in L^{2p}$ for some $p>n$. Assume also that $\psi\in W^{2,p}_{+}$. Define
$$\hat{g}=\psi^{\frac{4}{n-2}}g,~~~\hat{w}=\psi^{-N}w,~~~\hat{\tau}=\tau.$$
Then $u$ is a supersolution (resp. subsolution) to \eqref{Lichnerowicz} if and only if $\hat{u}=\psi^{-1}u$ is a supersolution (resp. subsolution) to the conformally transformed equation
\begin{equation}\label{Lichnerowicz2}
\frac{4(n-1)}{n-2}\Delta_{\hat{g}} \hat{u}+R_{\hat{g}}\hat{u}+\frac{n-1}{n}\hat{\tau}^{2}\hat{u}^{N-1}=\frac{\hat{w}^{2}}{\hat{u}^{N+1}}.
\end{equation}
In particular, $u$ is a solution to \eqref{Lichnerowicz} if and only if $\hat{u}$ is a solution to (\ref{Lichnerowicz2}).
\end{Lemma}
From the techniques in \cite{G-NQA}, we get the following remark.

\begin{Remark}\label{assumptionR}
Theorem \ref{maxwell1} guarantees that for any given $w \in L^{2p}\setminus \{0\}$, there exists a unique solution $u\in W^{2,p}_{+}$ to \eqref{Lichnerowicz}.
In addition, by direct calculation, we compute for any $k\geq N$
$$\int_{M}{\hat{u}^{k}dv_{\hat{g}}}=\int_{M}{\psi^{N-k}u^{k}dv_{g}}~~\mbox{and}~~\int_{M}{\hat{w}^{k}dv_{\hat{g}}}=\int_{M}{\psi^{N(1-k)}w^{k}dv_{g}},$$
where $(\hat{g},\hat{u},\hat{w})$ is as in Lemma \ref{maxwell2}. It follows that
$$(\max\psi)^{\frac{N-k}{k}}\|u\|_{L^{k}_{g}}\leq\|\hat{u}\|_{L^{k}_{\hat{g}}}\leq (\min\psi)^{\frac{N-k}{k}}\|u\|_{L^{k}_{g}}$$
and 
$$(\max\psi)^{\frac{N(1-k)}{k}}\|w\|_{L^{k}_{g}}\leq\|\hat{w}\|_{L^{k}_{\hat{g}}}\leq (\min\psi)^{\frac{N(1-k)}{k}}\|w\|_{L^{k}_{g}}.$$
Without loss of generality, we can assume moreover that $R>0$ or $R\equiv 0$ or $R=-\frac{n-1}{n}\tau^{2}$ depending on the sign of $\mathcal{Y}_{g}$
(in the case $\mathcal{Y}_{g}<0$, we refer to Case $3$ of Theorem \ref{maxwell1}). Under this assumption, it is also helpful to keep in mind that the term
$Ru^{k+1}+\frac{n-1}{n}\tau^{2}u^{k+N-1}$ is uniformly bounded from below for all positive functions $u\in L^{\infty}$ and all $k\geq 0$.
In fact, if $R\geq 0$, it is obvious that $Ru^{k+1}+\frac{n-1}{n}\tau^{2}u^{k+N-1}\geq 0$. If $R=-\frac{n-1}{n}\tau^{2}$, then
$\frac{n-1}{n}\tau^{2}u^{k+1}\left(u^{N-2}-1\right)\geq -\frac{n-1}{n}(\max|\tau|)^{2}$, which is our claim.
\end{Remark}

The following lemma will be used all along the paper.

\begin{Lemma}[Maximum principle]\label{maxpriw}
Assume that $v,~u$ are a supersolution (resp. subsolution) and a positive solution respectively to \eqref{Lichnerowicz} associated with a fixed $w$, then 
$$v\geq u~(\mbox{resp. $\leq$}).$$
In particular, assume $u_{0}$ (resp. $u_{1}$) is a positive solution to \eqref{Lichnerowicz} associated to $w=w_{0}$ (resp. $w_{1}$). Assume moreover $w_{0}\leq w_{1}$, then $u_{0}\leq u_{1}$.
\end{Lemma}
We give a simple proof of this fact based on Theorem \ref{maxwell1} (even if the proof of Theorem \ref{maxwell1} requires the maximum principle). Another proof,
independent of Theorem \ref{maxwell1}, can be found in \cite{DGH}. 
\begin{proof} 
We will prove the supersolution case, the remaining cases being similar. Assume that $v, u$ are respectively a supersolution and a positive solution to \eqref{Lichnerowicz}
associated to a fixed $w$. Since $u$ is a solution, it is also a subsolution. Hence, so is $tu$ for all constant $t\in (0,1]$. Since $\min v>0$,
we now take $t$ small enough s.t. $tu\leq v$. By Proposition \ref{method sub-super}, we then conclude that there exists a solution $u'\in W^{2,p}$ to $\eqref{Lichnerowicz}$
satisfying $tu\leq u'\leq v$. On the other hand, by uniqueness of positive solutions to \eqref{Lichnerowicz} given by Theorem \ref{maxwell1}, we obtain that $u=u'$, and
hence get the desired conclusion.
\end{proof}

\section{A New Proof for the Limit Equation}
In this section we show how Schaefer's fixed point theorem can be applied to give a simpler proof of the main result in \cite{DGH}. We first recall its statement
(see \cite[Theorem 3.4.8]{KCC} or \cite[Theorem 11.6]{DT}).

\begin{Theorem} (\textbf{Leray-Schauder's fixed point})\label{schaefer}
Let $X$ be a Banach space and assume that $T~:~X\times [0,1]\rightarrow X$ is a continuous compact mapping, satisfying $T(x,0)=0$ for all $x\in X$. If the set $K=\{x\in X|~~\exists
t\in[0,1]~\mbox{such that}~x=T(x,t)\}$ is bounded. Then $T=T(.,1)$ has a fixed point.
\end{Theorem}
\begin{Corollary}\label{corSchae}(\textbf{Schaefer's fixed point})
Assume that $T~:~X\rightarrow X$ is continuous compact and that the set
$$K=\{x\in X|~\exists t\in [0,1]~\mbox{such that}~x=tT(x)\}$$
is bounded. Then $T$ has a fixed point.
\end{Corollary}
We now state the main theorem in \cite{DGH} and give an  alternative proof.
\begin{Theorem}\label{theoremofDHG}
 Let data be given on $M$ as specified in \eqref{condintial} and assume that \eqref{condinitial2} holds. Furthermore, assume that $\tau>0$, then at least one of the following assertions is true
\begin{itemize}
\item The constraint equations \eqref{CE} admits a solution $(\varphi,W)$ with $\varphi>0$. Furthermore, the set of solutions $(\varphi,W)\in W_{+}^{2,p}\times W^{2,p}$ is compact.
\item There exists a nontrivial solution $W\in W^{2,p}$ to the limit equation
\begin{equation}\label{limit}
-\frac{1}{2}L^{*}LW=\alpha_{0}\sqrt{\frac{n-1}{n}}|LW|\frac{d\tau}{\tau}
\end{equation}
for some $\alpha_{0}\in (0,1].$
\end{itemize}
\end{Theorem}
Dahl-Gicquaud-Humbert's proof of this theorem in \cite{DGH} goes as follows: first, they apply Schauder's fixed point theorem to solve a subcritical system, that is a small perturbation of the
system \eqref{CE} where some exponent $N$ is replaced by $N-\epsilon$. This provides a sequence $(u_{\epsilon})$ of solutions to the subcritical system which is expected to converge to a solution of \eqref{CE}
when $\epsilon$ tends to $0$. A study of the sequence $(u_{\epsilon})$ shows that this actually happens when the limit equation \eqref{limit} has no non-trivial solution.\\

In the proof we present here, we show that Shaefer's fixed point theorem can be applied as soon as \eqref{limit} has no non-trivial solution, leading directly to the existence  of a solution to \eqref{CE}.
This simplifies the proof.

\begin{proof}[Proof of Theorem \ref{theoremofDHG}] Let $T$ be given as Section 2. Recall that $T$ is a continuous compact map from $L^{\infty}$ into itself and $T(\varphi)>0$ for all $\varphi\in L^{\infty}$
(see \cite[Lemma 2.2]{DGH}) . Set
$$S=\left\{\varphi\in L^{\infty}/~~\exists t\in[0,1]:~~\varphi=tT(\varphi)\right\}.$$

If $S$ is bounded, we get a solution to \eqref{CE} by Corollary \ref{corSchae}. If $S$ is not bounded, there exists an unbounded sequence $(\varphi_{i})$ in $L^{\infty}$ w.r.t. $t_{i}$ such that
\begin{subequations}\label{afterreplacing}
\begin{eqnarray}
\frac{4(n-1)}{n-2}\Delta \psi_{i}+R\psi_{i}&=&-\frac{n-1}{n}\tau^{2}\psi_{i}^{N-1}+|\sigma+LW_{i}|^{2}\psi_{i}^{-N-1}\\
-\frac{1}{2}L^{*}LW_{i}&=&\frac{n-1}{n}\varphi_{i}^{N}d\tau,
\end{eqnarray}
\end{subequations}
where $\psi_{i}=T(\varphi_{i})$ and $\varphi_{i}=t_{i}\psi_{i}$. We modify the main idea in \cite{DGH} to obtain the (non-trivial) solution to the limit equation. We set $\gamma_{i}=\|\psi_{i}\|_{\infty}$
and rescale $\psi_{i},~W_{i}$ and $\sigma$ as follows:
$$\widetilde{\psi}_{i}=\gamma_{i}^{-1}\psi_{i},~~\widetilde{W}_{i}=\gamma_{i}^{-N}W_{i},~~\widetilde{\sigma}_{i}=\gamma_{i}^{-N}\sigma.$$
It may be worth noticing that $\gamma_{i}=\|\psi_{i}\|_{\infty}=\frac{1}{t_{i}}\|\varphi_{i}\|_{\infty} \to \infty$ as $i\to\infty$. The system \eqref{afterreplacing}, with $\varphi_{i}$ replaced by
$t_{i}\psi_{i}$ in the vector equation, can be rewritten as
\begin{subequations}\label{rescale}
\begin{eqnarray}
 \frac{1}{\gamma_{i}^{N-2}}\left(\frac{4(n-1)}{n-2}\Delta\widetilde{\psi}_{i}+R\widetilde{\psi}_{i}\right)&=&-\frac{n-1}{n}\tau^{2}\widetilde{\psi}_{i}^{N-1}+|\widetilde{\sigma}+L\widetilde{W}_{i}|^{2}\widetilde{\psi}_{i}^{-N-1}\\
-\frac{1}{2}L^{*}L\widetilde{W}_{i}&=&\frac{n-1}{n}t^{N}_{i}\widetilde{\psi}^{N}_{i}d\tau.
\end{eqnarray}
\end{subequations}
Since $\|\widetilde{\psi}_{i}\|_{\infty}=1$, we conclude from the vector equation that $\left(\widetilde{W}_{i}\right)_{i}$ is bounded in $W^{2,p}$ and then by the Sobolev embedding,
(after passing to a subsequence) $\widetilde{W}_{i}$ converges in the $C^{1}$-norm to some $\widetilde{W}_{\infty}$. We now prove that 
\begin{equation}\label{converging}
\widetilde{\psi}_{i}\to\left(\sqrt{\frac{n}{n-1}}\frac{|L\widetilde{W}_{\infty}|}{\tau}\right)^{\frac{1}{N}}~~~\mbox{in $L^{\infty}$.}
\end{equation}
 Note that if such a statement is proven, passing to the limit in the vector equation, we see that $\widetilde{W}_{\infty}$ is a solution to the limit equation with (after passing to a subsequence)
 $\alpha_{0}=\lim t^{N}_{i}\in [0,1]$. On the other hand, since $\|\widetilde{\psi}_{i}\|_{\infty}=1$ for all $i$, $\widetilde{W}_{\infty}\nequiv 0$ from \eqref{converging} and then by the assumption
 that $(M,g)$ has no conformal Killing vector field, we obtain that $\alpha_{0}\ne 0$ which completes the proof.\\
 
 For any $\epsilon>0$, since $\frac{|L\widetilde{W}_{\infty}|}{\tau}\in C^{0}$, we can choose $\widetilde{\omega}\in C_{+}^{2}$ s.t.
 \begin{equation}\label{omegaandW}
 \biggl|\widetilde{\omega}-\left(\sqrt{\frac{n}{n-1}}\frac{|L\widetilde{W}_{\infty}|}{\tau}\right)^{\frac{1}{N}}\biggr|< \frac{\epsilon}{2}.
 \end{equation}
 To show \eqref{converging}, it suffices to show that
 $$|\widetilde{\psi}_{i}-\widetilde{\omega}|\leq \frac{\epsilon}{2}$$
 for all $i$ large enough. We argue by contradiction. Assume that the previous inequality is not true. We first consider the case when (after passing to a subsequence) there exists a sequence
 $(m_{i})\in M$ s.t. 
 \begin{equation}\label{omegaandpsi}
 \widetilde{\psi}_{i}(m_{i})>\widetilde{\omega}(m_{i})+\frac{\epsilon}{2}.
 \end{equation}
 By Lemma \ref{maxpriw} and Inequality (\ref{omegaandpsi}), $\widetilde{\omega}+\frac{\epsilon}{2}$ is not a supersolution to the rescaled Lichnerowicz equation. As a consequence, since $\Delta$ is here assumed to be the nonnegative Laplace, there exists a sequence
 $(p_{i})\in M$ satisfying
 \begin{align*}
 \frac{1}{\gamma_{i}^{N-2}}\left[\frac{4(n-1)}{n-2}\Delta\left(\widetilde{\omega}+\frac{\epsilon}{2}\right)(p_{i})+R\left(\widetilde{\omega}+\frac{\epsilon}{2}\right)(p_{i})\right]&+\frac{n-1}{n}\tau^{2}(p_{i})\left(\widetilde{\omega}+\frac{\epsilon}{2}\right)^{N-1}(p_{i})\\
 &<
 \left|\widetilde{\sigma}_{i}(p_{i})+L\widetilde{W}_{i}(p_{i})\right|^{2}\left(\widetilde{\omega}+\frac{\epsilon}{2}\right)^{-N-1}(p_{i}).
 \end{align*}
 Without loss of generality, we can assume that there exists $p_\infty \in M$ such that $p_i \to p_\infty$.
 Since $\left(\widetilde{\omega}+\frac{\epsilon}{2}\right)$ and $\tau$ are positive, the previous inequality can be rewritten as follows
 \begin{align*}
 \frac{n\left(\widetilde{\omega}+\frac{\epsilon}{2}\right)^{N+1}(p_{i})}{(n-1)\tau^{2}(p_{i})\gamma_{i}^{N-2}}\left[\frac{4(n-1)}{n-2}\Delta\left(\widetilde{\omega}+\frac{\epsilon}{2}\right)(p_{i})+R\left(\widetilde{\omega}+\frac{\epsilon}{2}\right)(p_{i})\right]&+\left(\widetilde{\omega}+\frac{\epsilon}{2}\right)^{2N}(p_{i})\\
 &<
 \frac{n}{n-1}\left|\widetilde{\sigma}_{i}(p_{i})+L\widetilde{W}_{i}(p_{i})\right|^{2}\tau^{-2}(p_{i}).
 \end{align*}
 Taking $i\to\infty$, due to the facts that $\widetilde{\omega}\in C^{2}_{+}$, $\min\tau>0$, $\gamma_{i}\to\infty$ and $\widetilde{W}_{i}\to\widetilde{W}_{\infty}$ in $C^{1}-$norm, we obtain that
 
$$\frac{n\left(\widetilde{\omega}+\frac{\epsilon}{2}\right)^{N+1}(p_{i})}{(n-1)\tau^{2}(p_{i})\gamma_{i}^{N-2}}\left[\frac{4(n-1)}{n-2}\Delta\left(\widetilde{\omega}+\frac{\epsilon}{2}\right)(p_{i})+R\left(\widetilde{\omega}+\frac{\epsilon}{2}\right)(p_{i})\right]\to 0,$$

  $$\left(\widetilde{\omega}+\frac{\epsilon}{2}\right)^{2N}(p_{i})\to \left(\widetilde{\omega}+\frac{\epsilon}{2}\right)^{2N}(p_{\infty})$$
and
$$\frac{n}{n-1}\left|\widetilde{\sigma}_{i}(p_{i})+L\widetilde{W}_{i}(p_{i})\right|^{2}\tau^{-2}(p_{i})\to \frac{n}{n-1}\left(\frac{|L\widetilde{W}_{i}|}{\tau}\right)^{2}(p_{\infty}),$$ 
This proves that 
$$\widetilde{\omega}(p_{\infty})+\frac{\epsilon}{2}\leq \left(\sqrt{\frac{n}{n-1}}\frac{|L\widetilde{W}_{\infty}|}{\tau}\right)^{\frac{1}{N}}(p_{\infty}),$$
which is a contradiction with (\ref{omegaandW}).\\

For the remaining case, i.e. when there exists a sequence $\left(m_{i}\right)\in M$ s.t.
$\widetilde{\omega}(m_{i})-\frac{\epsilon}{2}>\widetilde{\psi}_{i}(m_{i})$, $\widetilde{\omega}-\frac{\epsilon}{2}$ is not a subsolution to the rescaled Lichnerowicz equation on
$B_{\epsilon}=\left\{m\in M:~~\widetilde{\omega}(m)-\frac{\epsilon}{2}>0\right\}$
(here note that $\widetilde{\psi}_{i}>0$, then $\widetilde{\omega}(m_{i})-\frac{\epsilon}{2}>0$ and $\widetilde{\omega}-\frac{\epsilon}{2}<\widetilde{\psi}_{i}$ on
$\partial B_{\epsilon}$ if $B_{\epsilon}\subsetneq M$). By similar arguments to the first case, we also obtain a contradiction.
\end{proof}
 
The condition $\tau>0$ plays an important role in the proof of the main theorem in \cite{DGH} (or Theorem \ref{theoremofDHG}). Indeed, this condition implies that for any $(u,w)$ satisfying
\eqref{Lichnerowicz}, we have 
$$u^{N}\leq C(g,\tau,\sigma)\max\{\|w\|_{\infty},1\}$$ 
(it is a consequence of the maximum principle), which plays a crucial role in the proof. When $\tau$ vanishes, this inequality does not remain true as shown by the following proposition:

\begin{Proposition}\label{counterexample}
Let $\tau:~M\rightarrow \mathbb{R}$ be a $C^{0}$ function. For any $k>1$, we denote by $u_{k}>0$ the unique solution to \eqref{Lichnerowicz}
associated to $w=k$. Assume that $\tau$ vanishes somewhere, then
$$\frac{\|u_{k}\|_{\infty}^{N}}{k}\to \infty~~~\mbox{as $k\to\infty$}.$$
\end{Proposition}

\begin{proof}
Set $\widetilde{u}_{k}:=u_{k}/k^{\frac{1}{N}}$, then $\widetilde{u}_{k}$ is a solution to the following equation:
\begin{equation}\label{rescale equation}
\frac{1}{k^{\frac{N-2}{N}}}\left(\frac{4(n-1)}{n-2}\Delta \widetilde{u}_{k}+R\widetilde{u}_{k}\right)+\frac{n-1}{n}\tau^{2}\widetilde{u}_{k}^{N-1}=\frac{1}{\widetilde{u}_{k}^{N+1}}.
\end{equation}
Given $A>0$, we set
\begin{equation}\label{subsol}
\widetilde{\varphi}_{A}=\min\left\{\left(\frac{n}{(n-2)\tau^{2}}\right)^{\frac{1}{2N}}, A\right\}.
\end{equation}
Fix $\epsilon>0$ small enough, we first prove that 
\begin{equation}\label{*}
\widetilde{\varphi}_{A}\leq \widetilde{u}_{k}+2\epsilon, ~~~\forall k\geq k_{A},
\end{equation}
for some $k_{A}$ large enough depending on $A$. We proceed by contradiction. Assume that this is not true, so there exists a subsequence $\{m_{k}\}\in M$ s.t.
\begin{equation}\label{subA}
\widetilde{\varphi}_{A}(m_{k})-2\epsilon>\widetilde{u}_{k}(m_{k}).
\end{equation}
Next since $\widetilde{\varphi}_{A}\in C^{0}_{+}$, we can choose $\widetilde{\phi}_{A}\in C^{2}_{+}$ s.t.
\begin{equation}\label{regularC2}
|\widetilde{\phi}_{A}-\widetilde{\varphi}_{A}|\leq \epsilon/2.
\end{equation}  
Then it follows from \eqref{subA} that
\begin{equation}\label{subA2}
\widetilde{\phi}_{A}(m_{k})-\epsilon>\widetilde{u}_{k}(m_{k}).
\end{equation}
Set $B_{A}=\left\{m\in M:~ \widetilde{\phi}_{A}-\epsilon>0\right\}$. Since $\widetilde{u}_{k}>0$, we deduce from \eqref{subA2} that $\widetilde{\phi}_{A}-\epsilon$
is not a subsolution to \eqref{rescale equation} and hence there exists a sequence $\{p_{k}\}\in B_{A}$ s.t.
$$\frac{1}{k^{\frac{N-2}{N}}}\left[\frac{4(n-1)}{n-2}\Delta \left(\widetilde{\phi}_{A}-\epsilon\right)(p_{k})+R(p_{k})\left(\widetilde{\phi}_{A}-\epsilon\right)(p_{k})\right]+\frac{n-1}{n}\tau^{2}(p_{k})\left(\widetilde{\phi}_{A}-\epsilon\right)^{N-1}(p_{k})>\frac{1}{\left(\widetilde{\phi}_{A}-\epsilon\right)^{N+1}(p_{k})}$$
or equivalently,
$$\frac{\left(\widetilde{\phi}_{A}-\epsilon\right)^{N+1}(p_{k})}{k^{\frac{N-2}{N}}}\left[\frac{4(n-1)}{n-2}\Delta \left(\widetilde{\phi}_{A}-\epsilon\right)(p_{k})+R(p_{k})\left(\widetilde{\phi}_{A}-\epsilon\right)(p_{k})\right]+\frac{n-1}{n}\tau^{2}(p_{k})\left(\widetilde{\phi}_{A}-\epsilon\right)^{2N}(p_{k})>1.$$
Taking $k\to \infty$ and assuming (after passing to a subsequence) $p_i \to p_\infty$, we obtain that
$$\frac{\left(\widetilde{\phi}_{A}-\epsilon\right)^{N+1}(p_{k})}{k^{\frac{N-2}{N}}}\left[\frac{4(n-1)}{n-2}\Delta \left(\widetilde{\phi}_{A}-\epsilon\right)(p_{k})+R(p_{k})\left(\widetilde{\phi}_{A}-\epsilon\right)(p_{k})\right]\to 0$$
and
$$\frac{n-1}{n}\tau^{2}\left(\widetilde{\phi}_{A}-\epsilon\right)^{2N}(p_{k})\to \frac{n-1}{n}\tau^{2}(p_{\infty})\left(\widetilde{\phi}_{A}-\epsilon\right)^{2N}(p_{\infty}),$$
This shows that
\begin{equation}\label{counterexample2}
\frac{n-1}{n}\tau^{2}(p_{\infty})\left(\widetilde{\phi}_{A}-\epsilon\right)^{2N}(p_{\infty})\geq 1.
\end{equation}
On the other hand, we have
\begin{align*}
\frac{n-1}{n}\tau^{2}(p_{\infty})\left(\widetilde{\phi}_{A}-\epsilon\right)^{2N}(p_{\infty})&\leq \frac{n-1}{n}\tau^{2}(p_{\infty})\left(\widetilde{\varphi}_{A}-\frac{\epsilon}{2}\right)^{2N}(p_{\infty})\qquad\mbox{(by \eqref{regularC2})}\\
&\leq \frac{n-1}{n}\tau^{2}(p_{\infty})\left(\widetilde{\varphi}_{A}^{2N}(p_{\infty})-\left(\frac{\epsilon}{2}\right)^{2N}\right)\\
&\leq 1-\frac{n-1}{n}\tau^{2}(p_{\infty})\left(\frac{\epsilon}{2}\right)^{2N}\\
&<1,
\end{align*}
which is a contradiction with \eqref{counterexample2}, and then \eqref{*} holds, as claimed. Now if $\widetilde{u}_{k}\leq C$, we deduce  from
\eqref{*} that $\max\widetilde{\varphi}_{A}\leq C+2\epsilon$, which is false when $A\to +\infty$ since $\tau$ has some zeros. The proof is completed.
 \end{proof}
 
We can be more precise. This is the content of the next proposition
\begin{Proposition}\label{estimateLp}
Let $\tau:~M\rightarrow \mathbb{R}$ be a $C^{0}$ function. We set
$$L=\left\{(u,w)\in W_{+}^{2,p}\times L^{\infty}:~(u,v)~\mbox{satisfies \eqref{Lichnerowicz}}\right\}.$$ Given $\alpha\geq\frac{1}{N},$   
$\sup_{(u,w)\in L}{\frac{\|u\|_{L^{N\alpha}}^{N}}{\max\{\|w\|_{\infty},1\}}}$ is bounded if and only if  $|\tau|^{-\alpha}\in L^{1}$.
\end{Proposition}
\begin{proof} Applying Lemma \ref{maxpriw} with $w_{0}=w$ and $w_{1}=\|w\|_{\infty}$, we have
\begin{align*}
\sup_{(u,w)\in L}{\frac{\|u\|_{L^{N\alpha}}^{N}}{\max\{\|w\|_{\infty},1\}}}&=\sup_{\substack{(u,w)\in L\\w\mbox{ {\scriptsize constant}}}}{\frac{\|u\|_{L^{N\alpha}}^{N}}{\max\{|w|,1\}}}\\
&=\sup_{k>1}{\frac{\|u_{k}\|^{N}_{L^{N\alpha}}}{k}},
\end{align*}
where $u_{k}$ is the unique positive solution to \eqref{Lichnerowicz} associated to $w=k$. Therefore,
$\sup_{(u,w)\in L}{\frac{\|u\|_{L^{N\alpha}}^{N}}{\max\{\|w\|_{\infty},1\}}} < \infty$ if and only if $\frac{\|u_{k}\|^{N}_{L^{N\alpha}}}{k}$
is uniformly bounded for all $k>1$. Moreover note that with $C=C(g,\tau)$ large enough and not depending on $k$, $k^{\frac{1}{N}}/C$ is a subsolution
to \eqref{Lichnerowicz} associated to $w=k$, and hence for all $k>1$,
\begin{equation}\label{*++}
u_{k}\geq \frac{k^{\frac{1}{N}}}{C}>\frac{1}{C}
\end{equation}
We first prove that $|\tau|^{-\alpha}\in L^{1}$ is a necessary condition. Set $\widetilde{u}_{k}=u_{k}/k^{\frac{1}{N}}$ and we let $\widetilde{\varphi}_{A}$ given by \eqref{subsol}.
As in the proof of Proposition \ref{counterexample}, we obtain that for all $k$ large enough and depending on $A$,
 $$\widetilde{\varphi}_{A}\leq \widetilde{u}_{k}+\epsilon.$$
Assume that $\widetilde{u}_{k}$ is uniformly bounded in $L^{N\alpha}$, so is $\widetilde{\varphi}_{A}$ by the previous inequality. On the other hand, it is clear that $\widetilde{\varphi}_{A}$
converges pointwise a.e to $\left(\frac{n}{n-1}\right)^{\frac{1}{2N}}|\tau|^{-\frac{1}{N}}$ as $A\to\infty$. Hence the monotone convergence theorem implies that $|\tau|^{-\frac{1}{N}}\in L^{N\alpha}$,
which is our claim.\\

We now prove that the condition is sufficient. Assume that $|\tau|^{-\alpha}\in L^{1}$. Multiplying \eqref{Lichnerowicz} by $u_{k}^{N\alpha+N+1}$ and integrating over $M$, we have
\begin{equation}\label{equationwithalpha}
\frac{4(n-1)}{n-2}\int_{M}{u_{k}^{N\alpha+N+1}\Delta u_{k}dv}+\int_{M}{Ru_{k}^{N\alpha+N+2}dv}+\frac{n-1}{n}\int_{M}{\tau^{2}u_{k}^{N(\alpha+2)}dv}=k^{2}\int_{M}{u_{k}^{N\alpha}dv}.
\end{equation}
As observed in Remark \ref{assumptionR},  $Ru_{k}^{N\alpha+N+2}+\frac{n-2}{n}\tau^{2}u_{k}^{N(\alpha+2)}$ is uniformly bounded from below by a constant $\zeta=\zeta(g,\tau)$ which does not depend on $k$ since we assume that $R\geq 0$ or $R=-\frac{n-1}{n}\tau^{2}$. Moreover, we have
$$\int_{M}{u_{k}^{N\alpha+N+1}\Delta u_{k}dv}=\frac{N\alpha+N+1}{(\frac{N\alpha+N}{2}+1)^{2}}\int_{M}{|\nabla\varphi^{\frac{N\alpha+N}{2}+1}|^{2}dv}.$$
These facts combined with \eqref{*++}-\eqref{equationwithalpha} lead to
\begin{equation}\label{**}
\int_{M}{\tau^{2}u_{k}^{N(\alpha+2)dv}}\leq C_{1}(C,\zeta)k^{2}\int_{M}{u_{k}^{N\alpha}dv},
\end{equation}
 On the other hand, we get that
\begin{align*}
\int_{M}{u_{k}^{N\alpha}dv}&\leq \left(\int_{M}{|\tau|^{-\alpha}dv}\right)^{\frac{2}{\alpha+2}}\left(\int_{M}{\tau^{2}u_{k}^{N(\alpha+2)}dv}\right)^{\frac{\alpha}{\alpha+2}}~~(\mbox{by H\"{o}lder inequality})\\
&\leq C_{2}(C_{1},\tau,\alpha)\left(k^{2}\int_{M}{u_{k}^{N\alpha}dv}\right)^{\frac{\alpha}{\alpha+2}}~~~~~(\mbox{by (\ref{**})}).
\end{align*}
 It follows easily that for all $k>1$
$$\frac{\|u_{k}\|_{L^{N\alpha}}^{N}}{k}\leq C^{\frac{\alpha+2}{2\alpha}}_{2},$$ 
which completes our proof.
\end{proof}
The fixed point theorem above has some other consequences that we describe now. First, we have the following proposition.

\begin{Proposition}\label{remark 1}
 Let data be given on $M$ as specified in \eqref{condintial} and assume that $(M,g)$ has no conformal Killing vector field and $\sigma\nequiv 0$. If $\mathcal{Y}_{g}>0$, then there exists a constant $\alpha=\alpha(g,\tau,\sigma)\in (0,1]$ such that the constraint equations w.r.t. the new data $(g,\alpha\tau,\sigma)$ admits a solution.
\end{Proposition}

\begin{Remark}
In the proof, we apply Leray-Schauder's Theorem \ref{schaefer} and not Corollary \ref{corSchae} as in the proof of Theorem \ref{theoremofDHG}.
\end{Remark}

\begin{proof}
By Remark \ref{assumptionR}, we may assume $R>0$. We construct a compact map $\widetilde{T}: L^{\infty}\times [0,1]\rightarrow L^{\infty}$ as follows. For each
$(\varphi,t)\in L^{\infty}\times [0,1]$, there exists a unique $W_\varphi\in W^{2,p}$ s.t.
$$-\frac{1}{2}L^{*}LW_\varphi=\frac{n-1}{n}\varphi^{N}d\tau$$
and there exists a unique positive $\psi\in W^{2,p}$ satisfying 
$$\frac{4(n-1)}{n-2}\Delta \psi+R\psi=-\frac{n-1}{n}t^{2N}\tau^{2}\psi^{N-1}+|\sigma+LW_\varphi|^{2}\psi^{-N-1}$$
(see \cite[Lemma 2.2]{DGH} and notice that $R>0$).
Then we define
$$\widetilde{T}(\varphi,t)=t\psi.$$
The continuity and compactness of $\widetilde{T}$ is clearly a direct consequence of the continuity and compactness of $T'(\varphi,t)\coloneqq \frac{\widetilde{T}(\varphi,t)}{t}=\psi$.\\

Note that $T'(\varphi,t)=\widetilde{T}_{1}(G(\varphi),t)$. Here $G(\varphi)=|LW_\varphi+\sigma|\nequiv 0$ and $\widetilde{T}_{1}:~L^{\infty}\times [0,1]\rightarrow W_{+}^{2,p}$ is defined by
$\widetilde{T}_{1}(w,t)=\psi$, where

\begin{equation}\label{Line.res.t}
\frac{4(n-1)}{n-2}\Delta \psi+R\psi=-\frac{n-1}{n}t^{2N}\tau^{2}\psi^{N-1}+w^{2}\psi^{-N-1}.
\end{equation}

As proven in \cite{DGH}, $G$ is continuous compact, so the continuity and compactness of $T'$ and hence that of $\widetilde{T}$, will follow from the continuity of
$\widetilde{T}_{1}$. Actually, we prove more: $\widetilde{T}_{1}$ is a $C^{1}-$map.  Indeed, define $F:L^{\infty}\times [0,1]\times W_{+}^{2,p}\rightarrow L^{2p}$ by
$$F(w,t,\psi)= \frac{4(n-1)}{n-2}\Delta \psi+R\psi+\frac{n-1}{n}t^{2N}\tau^{2}\psi^{N-1}-w^{2}\psi^{-N-1}.$$
It is clear that $F$ is continuous and $F(w,t,\widetilde{T}_{1}(w,t))=0$ for all $(w,t)\in L^{\infty}\times [0,1]$.
A standard computation shows that the Fr\'{e}chet derivative of $F$ w.r.t. $\psi$ is given by
$$F_{\psi}(w,t)(u)=\frac{4(n-1)}{n-2}\Delta u+Ru+\frac{(N-1)(n-1)}{n}t^{2N}\tau^{2}\psi^{N-2}u+(N+1)w^{2}\psi^{-N-2}u.$$
We first note that $F_{\psi}\in C\left(L^{\infty}\times [0,1], L(W^{2,p},L^p)\right)$, where $L(W^{2,p}, L^p)$ denotes the Banach space of all linear continuous maps from $W^{2,p}$ into $L^p$.
Now, given $(w_{0},t_{0})\in L^{\infty}\times [0,1]$, setting $\psi_{0}=\widetilde{T}_{1}(w_{0},t_{0})$, we have
$$F_{\psi_{0}}(w_{0},t_{0})(u)=\frac{4(n-1)}{n-2}\Delta u+\left(R+\frac{(N-1)(n-1)}{n}t_{0}^{2N}\tau_{0}^{2}\psi_{0}^{N-2}+(N+1)w_{0}^{2}\psi_{0}^{-N-2}\right)u.$$
Since 
$$R+\frac{(N-1)(n-1)}{n}t_{0}^{2N}\tau_{0}^{2}\psi_{0}^{N-2}+(N+1)w_{0}^{2}\psi_{0}^{-N-2}\geq \min R>0,$$
we conclude that $F_{\psi_{0}}(w_{0},t_{0}):~W^{2,p}\to L^p$ is an isomorphism. The implicit function theorem then implies that $\widetilde{T}_{1}$ is a $C^{1}$ function in a
neighborhood of $(w_{0},t_{0})$, which proves our claim.\\

Next applying Leray-Schauder's Theorem \ref{schaefer} to $\widetilde{T}$, we obtain as a direct consequence that there exist $\varphi_{0}\in L^{\infty}$ and $t_{0}\in (0,1]$ s.t.
\begin{subequations}
\begin{eqnarray*}
\frac{4(n-1)}{n-2}\Delta \psi_{0}+R\psi_{0}&=&-\frac{n-1}{n}t^{2N}_{0}\tau^{2}\psi_{0}^{N-1}+|\sigma+LW_{0}|^{2}\psi_{0}^{-N-1}\\
-\frac{1}{2}L^{*}LW_{0}&=&\frac{n-1}{n}\varphi_{0}^{N}d\tau,
\end{eqnarray*}
\end{subequations}
with $\varphi_{0}=t_{0}\psi_{0}\in W^{2,p}$. Indeed, set
$$K=\left\{\varphi\in L^{\infty}|~~\exists
t\in[0,1]~\mbox{such that}~\varphi=\widetilde{T}(\varphi,t)\right\}.$$
It is clear that $\widetilde{T}(\varphi,0)=0$ for all $\varphi\in L^{\infty}$. Assume that such $(\varphi_{0},t_{0})$ does not exist. Then $K=\left\{0\right\}$. By Leray-Schauder's Theorem
\ref{schaefer}, there exists $\varphi$ s.t. $\varphi=\widetilde{T}(\varphi,1)=T(\varphi)$ which belongs to $K$. So $\varphi=0$ which is impossible since $T(\varphi)\nequiv 0$.\\

Now replacing $\varphi_{0}$ by $t_{0}\psi_{0}$ in the vector equation, we get that $(\psi_{0},W_{0})$ is a solution to \eqref{CE} w.r.t. the new data $(g,\alpha\tau,\sigma)$, with $\alpha=t_{0}^{N}$.
\end{proof}

Proposition \ref{remark 1} is a direct consequence of the small-TT case (i.e. a smallness assumption on the transverse-traceless tensor) in \cite{HNT} and \cite{Maxwell}. More precisely, we can easily
check the following, which is developed further in \cite{G-NQA}.

\begin{Remark}\label{remark2}
$(\varphi,W)$ is a solution to the constraint equations w.r.t. an initial data $(g,\tau,\sigma)$ if and only if  $( C^{-1}\varphi,C^{-\frac{N+2}{2}}W)$ is a solution to the constraint equation w.r.t.
the data $(g,C^{\frac{N-2}{2}}\tau, C^{-\frac{N+2}{2}}\sigma)$ for all constant $C>0$.
\end{Remark}

\begin{Proposition}\label{farCMCofMawell} (see \cite{HNT} or \cite{Maxwell})
Let data be given on $M$ as specified in \eqref{condintial}. Assume that $\mathcal{Y}_{g}>0$, $(M,g)$ has no conformal Killing vector field and $\sigma\nequiv 0$. If $\|\sigma\|_{L^{\infty}}$ is small
enough (only depending on $g$ and $\tau$), then the system \eqref{CE} has a solution $(\varphi,W)$.
\end{Proposition}

\noindent From Remark \ref{remark2}, with $C=\alpha^{-\frac{2}{N-2}}$, Proposition \ref{remark 1} is equivalent to the fact that \eqref{CE} w.r.t. the new data $(g,\tau,\alpha^{\frac{N+2}{N-2}}\sigma)$
admits a solution, and this holds for $\alpha$ small enough by Proposition \ref{farCMCofMawell}.\\

In particular, this approach has the advantage to give an unifying point of view of the limit equation method in \cite{DGH} and the far-from CMC results in  \cite{G-NQA}, \cite{HNT} and \cite{Maxwell}.\\

The main theorem in \cite{DGH} (or Theorem \ref{theoremofDHG}) says that the non-existence of non-trivial solution to the limit equation \eqref{limit} implies the existence of a solution to
\eqref{CE}. The opposite question naturally arises whether the existence of a solution to \eqref{CE} implies the non-existence of (non-trivial) solution to the limit equation. The following
proposition shows that this is false.

\begin{Proposition}\label{remark3}
There exists an initial data $(M,g,\tau,\sigma)$ such that both the corresponding \eqref{CE} and \eqref{limit} admit (non-trivial) solutions. 
\end{Proposition}

\begin{proof} In \cite{DGH}, Dahl-Gicquaud-Humbert prove that there exist $(M,g,\tau,\sigma)$ and $\alpha_{0}\in (0,1]$ s.t. $\mathcal{Y}_{g}>0$ and the corresponding limit equation
$$-\frac{1}{2}L^{*}LW=\alpha_{0}\sqrt{\frac{n-1}{n}}|LW|\frac{d\tau}{\tau}$$
admits a nontrivial solution $W\in W^{2,p}$ (see \cite[Proposition 1.6]{DGH}). Now note that for all $\alpha>0$,
$$\frac{d\alpha\tau}{\alpha\tau}=\frac{d\tau}{\tau}.$$
so the limit equation for the 4-tuple $(M,g,\alpha\tau,\sigma)$ also admits a non-trivial solution.
Taking $\alpha$ given by Proposition \ref{remark 1} provides $(M,g,\alpha\tau,\sigma)$ as desired.
\end{proof}

\section{Half-Continuous Maps and Applications}
In this section we introduce the theory of half-continuous functions and its applications to solve the constraint equations.
We summarize results on half-continuous maps in the next subsection. For the proofs we refer the reader to \cite{PB} or \cite{TK}.
\subsection{Half-Continuous Maps}
\begin{Definition}
Let $C$ be a subset of a Banach space $X$. A map $f~:~C \rightarrow X$ is said to be half-continuous if for each $x\in C$ with $x\ne f(x)$ there exists $p\in X^{*}$ and a neighborhood $W$ of $x$ in $C$ such that
$$\langle p,f(y)-y\rangle >0$$
for all $y\in W$ with $y\ne f(y)$.
\end{Definition} 
The following proposition gives a relation between  half-continuity and continuity.

\begin{Proposition}[see \cite{TK}, Proposition 3.2]
Let $X$ be a Banach space. Then every continuous map $f~:~C\rightarrow X$ is half-continuous.
\end{Proposition}

\begin{Remark}[see \cite{TK}]
There are some half-continuous maps which are not continuous. For example, let $f~:~\mathbb{R}\rightarrow \mathbb{R}$ be defined by
$$
f(x)=\left\{\begin{array}{ll}
3&\textrm{if $x\in[0,1)$,}\\
2&\textrm{otherwise.}
\end{array} \right.
$$
Then $f$ is half-continuous but not continuous.
\end{Remark}

\begin{Theorem}[see \cite{TK}, Theorem 3.9 or \cite{PB}, Theorem 3.1]\label{fixed point theorem}
Let $C$ be a nonempty compact convex subset of a Banach space $X$. If $f~:~C\rightarrow C$ is half-continuous, then $f$ has a fixed point. 
\end{Theorem}

A direct consequence of Theorem \ref{fixed point theorem} is the following corollary, which is our main tool in the next subsection.

\begin{Corollary}\label{Fixed point theorem}
Let $C$ be a nonempty closed convex subset of a Banach space $X$. If $f:C\rightarrow ~C$ is half-continuous and $f(C)$ is precompact, then $f$ has a fixed point.
\end{Corollary}

\begin{proof}
Since $\overline{f(C)}$ is nonempty compact and $X$ is a Banach space, $\overline{\text{conv}}(f(C))$ is a nonempty compact convex subset of $X$
(see \cite{Rudin}, Theorem 3.20). Moreover, since $C$ is a closed convex subset of $X$ and $f(C)\subset C$, we have $\overline{\text{conv}}(f(C))\subset C,$
and hence $f\left(\overline{\text{conv}}(f(C))\right)\subset f(C)\subset \overline{\text{conv}}(f(C)).$ Now restricting $f$ to $\overline{\text{conv}}(f(C))$
and applying the previous theorem, we obtain the desired conclusion. 
\end{proof}

\subsection{Existence Results for Modified Constraint Equations}
Here we apply the concept of half-continuity to improve recent existence results for \eqref{CE} (see \cite{HNT} or \cite{Maxwell}).\\

The first non-CMC result for \eqref{CE} is  the near-CMC case, which is presented by many authors: if $\frac{\max{|d\tau|}}{\min{|\tau|}}$ is small enough,
then \eqref{CE} admits a solution (see \cite{BI}). Recently, Dahl-Gicquaud-Humbert \cite{DGH} improved this result. They show that \eqref{CE} has a solution,
provided $\|\frac{d\tau}{\tau}\|_{L^{n}}$ is small enough (see \cite[Corollary 1.3 and 14]{DGH}). However, for a smooth vanishing $\tau$, these assumptions never
hold.  Therefore, we treat a generalization of \eqref{CE}, with $d\tau$ replaced by a $1$-form $\xi\in L^{\infty}$ in the vector equation. Namely, let data be
given on $M$ as specified in \eqref{condintial} and choose also a $1-$form $\xi\in L^{\infty}$. We are interested in the following system.

\begin{subequations}\label{Modified CE}
\begin{align}
\frac{4(n-1)}{n-2}\Delta \varphi+R\varphi &= -\frac{n-1}{n}\tau^{2}\varphi^{N-1}+|\sigma+LW|^{2}\varphi^{-N-1}\\
-\frac{1}{2}L^{*}LW &=\frac{n-1}{n}\varphi^{N}\xi.
\end{align}
\end{subequations}

Note that all the methods described above apply in this context when $\tau>0$. A natural question is then whether this coupled nonlinear elliptic system has
a solution under a similar condition, i.e. $\|\frac{\xi}{\tau}\|_{L^{n}}$ is small enough. As $\tau$ vanishes, it becomes more complicated to apply the method
of global supersolution introduced by Holst-Nagy-Tsogtgerel \cite{HNT} because the construction of a supersolution to the Lichnerowicz equation fails with their
method near $Z(\tau)$, which is the zero set of $\tau$. Before going further, we establish a useful estimate for \eqref{CE}.\\

Let $\mathcal{I}$ be the family of all solutions of \eqref{CE} for fixed given data $(g,\tau,\sigma)$. Provided $\tau>0$, it was obtained in \cite{DGH} by
induction that there exists a positive constant $C=C(M,g,\tau,\sigma)$ s.t.
$$\|\varphi\|_{\infty}\leq C\max\{\|LW\|^{\frac{1}{N}}_{L^{2}},1\},~~\forall (\varphi,W)\in\mathcal{I}.$$ 
For a vanishing $\tau$, there is no reason to get the estimate above. However, by a slight change in the proof, we have the following proposition.
 
\begin{Proposition}\label{Estimate1 L2}
Let data be given on $M$ as specified in \eqref{condintial} and assume that \eqref{condinitial2} holds. Given $l>0$, there exists a positive constant
$C=C(M,g,\sigma,\tau,l)$ s.t. for any $(\varphi,W)\in \mathcal{I}$ satisfying $\|LW\|_{L^{2}}\leq l$ we have
$$\|\varphi\|_{\infty}\leq C.$$
Moreover, if $\mathcal{Y}_{g}>0$, the assumption that $Z(\tau)$ has zero Lebesgue measure can be omitted.
\end{Proposition}

\begin{proof} For simplicity, we assume that $\tau\in C^{1}(M)$. We begin with the observation that, to prove the proposition, it suffices to show that
there exists a constant $c=c(n,g,\tau,\sigma,l)>0$ s.t. for any $(\varphi,W)\in\mathcal{I}$ satisfying $\|LW\|_{L^{2}}\leq l$ we have $\|LW\|_{\infty}<c$.
In fact, assume that this is true. Then, from Lemma \ref{maxpriw}, we have that $\varphi\leq \varphi_{c}$, where $\varphi_{c}$ is a unique
positive solution to the Lichnerowicz equation \eqref{Lichnerowicz} associated to $w=c+\|\sigma\|_{\infty}$, and hence taking $C=\max\varphi_{c}$, the proposition follows.\\

Now we will prove the boundedness of $\|LW\|_{\infty}$ as mentioned above. Set $q_{i}=2\left(\frac{N+2}{4}\right)^{i}$ for all $i\in\mathbb{N}$, we first show inductively that
if $|LW|$ is uniformly bounded in $L^{q_{i}}-$norm by $r_{i}>0$, then so $|LW|$ is in $L^{q_{i+1}}$ by $r_{i+1}=r_{i+1}(n,g,\tau,\sigma,q_{i},r_{i})>0$. In fact, multiplying
the Lichnerowicz equation by $\varphi^{\frac{(N+2)q_{i}}{2}-1}$ and integrating over $M$, we have

\begin{equation}\label{multip}
\begin{aligned}
\frac{4(n-1)}{n-2}\int_{M}{\varphi^{\frac{(N+2)q_{i}}{2}-1}\Delta\varphi dv}
+\int_{M}{R\varphi^{\frac{(N+2)q_{i}}{2}} dv}+&\frac{n-1}{n}\int_{M}{\tau^{2}\varphi^{N+\frac{(N+2)q_{i}}{2}-2}dv}\\
&=\int_{M}{|\sigma+LW|^{2}\varphi^{\frac{(N+2)(q_{i}-2)}{2}}dv}\\
&\leq\|\sigma+LW\|_{L^{q_{i}}}^{2}\left(\int_{M}{\varphi^{\frac{(N+2)q_{i}}{2}}dv}\right)^{\frac{q_{i}-2}{q_{i}}}\\
&\qquad\mbox{(by $q_{i}\geq 2$ and H\"{o}lder inequality)}\\
&\leq 2\left(\|\sigma\|_{L^{q_{i}}}^{2}+\|LW\|_{L^{q_{i}}}^{2}\right)\left(\int_{M}{\varphi^{\frac{(N+2)q_{i}}{2}}dv}\right)^{\frac{q_{i}-2}{q_{i}}}.
\end{aligned}
\end{equation}
Since
\begin{equation}\label{Greenform}
\int_{M}{\varphi^{\frac{(N+2)q_{i}}{2}-1}\Delta \varphi dv}=\frac{8\left((N+2)q_{i}-2\right)}{(N+2)^{2}q_{i}^{2}}\int_{M}{|\nabla \varphi^{\frac{(N+2)q_{i}}{4}}|^{2}dv}\geq0,
\end{equation}
and since the term $\int_{M}{R\varphi^{\frac{(N+2)q_{i}}{2}} dv}+\frac{n-1}{n}\int_{M}{\tau^{2}\varphi^{N+\frac{(N+2)q_{i}}{2}-2}dv}$ is uniformly bounded from below as observed in Remark \ref{assumptionR}. we obtain from (\ref{multip}) that
$$\int_{M}{|\nabla\varphi^{\frac{(N+2)q_{i}}{4}}|^{2}dv}\leq c_{1}(g,\tau)+c_{2}(g,\tau,\sigma,r_{i})\left(\int_{M}{\varphi^{\frac{(N+2)q_{i}}{2}}dv}\right)^{\frac{q_{i}-2}{q_{i}}},$$
and then 
\begin{equation}\label{L2.0}
\begin{aligned}
\|\varphi^{\frac{(N+2)q_{i}}{4}}\|_{L^{N}}^{2}&\leq c_{3}(M,g)\left(\|\nabla\varphi^{\frac{(N+2)q_{i}}{4}}\|_{L^{2}}^{2}+\|\varphi^{\frac{(N+2)q_{i}}{4}}\|_{L^{2}}^{2}\right)~~\mbox{(by the Sobolev inequality)}\\
&\leq c_{3}\left(c_{1}+c_{2}\|\varphi^{\frac{(N+2)q_{i}}{4}}\|_{L^{2}}^{\frac{2(q_{i}-2)}{q_{i}}}+\|\varphi^{\frac{(N+2)q_{i}}{4}}\|_{L^{2}}^{2}\right).
\end{aligned}
\end{equation}
To show that $\|\varphi^{\frac{(N+2)q_{i}}{4}}\|_{L^{N}}$ is bounded, by \eqref{L2.0} it suffices to assume that
\begin{equation}\label{L2}
\|\varphi^{\frac{(N+2)q_{i}}{4}}\|_{L^{N}}\leq 3c_{3}\|\varphi^{\frac{(N+2)q_{i}}{4}}\|_{L^{2}}
\end{equation}
 and to prove that $\|\varphi^{\frac{(N+2)q_{i}}{4}}\|_{L^{2}}$ is bounded. We study  two cases.
 \begin{itemize}
 \item{Case 1. \textit{$\mathcal{Y}_{g}>0$:}} By Remark \ref{assumptionR}, we can assume that $R>0$ and then it is clear from \eqref{multip}-\eqref{Greenform} that
 \begin{align*}
 \int_{M}{\varphi^{\frac{(N+2)q_{i}}{2}}dv}&\leq \frac{2}{\min R}\left(\|\sigma\|_{L^{q_{i}}}^{2}+r_{i}^{2}\right)\left(\int_{M}{\varphi^{\frac{(N+2)q_{i}}{2}}dv}\right)^{\frac{q_{i}-2}{q_{i}}},\\
 \end{align*}
 which implies the boundedness of $\|\varphi^{\frac{(N+2)q_{i}}{4}}\|_{L^{2}}$.\\
 
\item{Case 2. \textit{$\mathcal{Y}_{g}\leq 0$:}}
 Given $k>0$, we define
$$B_{k}=\left\{m\in M:~~\varphi^{\frac{(N+2)q_{i}}{4}}(m)\geq \frac{1}{k}\|\varphi^{\frac{(N+2)q_{i}}{4}}\|_{L^{2}}\right\}.$$
Let $\chi_{B_{k}}$ denote the characteristic function of $B_{k}$. We have
\begin{align*}
1=\int_{M}{\frac{\varphi^{\frac{(N+2)q_{i}}{2}}}{\|\varphi^{\frac{(N+2)q_{i}}{4}}\|_{L^{2}}^{2}}dv}&\leq \int_{M}{\frac{\chi_{B_{k}}\varphi^{\frac{(N+2)q_{i}}{2}}}{\|\varphi^{\frac{(N+2)q_{i}}{4}}\|_{L^{2}}^{2}}dv}+\int_{M\setminus B_{k}}{\frac{\varphi^{\frac{(N+2)q_{i}}{2}}}{\|\varphi^{\frac{(N+2)q_{i}}{4}}\|_{L^{2}}^{2}}dv}\\
&\leq \frac{\|\varphi^{\frac{(N+2)q_{i}}{4}}\|_{L^{N}}^{2}}{\|\varphi^{\frac{(N+2)q_{i}}{4}}\|_{L^{2}}^{2}} Vol(B_{k})^{\frac{N-2}{N}}+\frac{1}{k^{2}}Vol(M\setminus B_{k})\\
&\qquad(\mbox{by H\"{o}lder inequality and the definition of $B_{k}$})\\
&\leq 9c_{3}^{2}Vol(B_{k})^{\frac{N-2}{N}}+\frac{1}{k^{2}}Vol(M)~~~~(\mbox{by \eqref{L2}}).
\end{align*}
Taking $k_{0}\geq 2Vol(M)+1$, it follows that
 $Vol(B_{k_{0}})\geq 2c_{4}(n,c_{3})>0$. On the other hand, since $Z(\tau)$ is a closed, zero-measurable subset of $M$, there exists a neighborhood $B_{i}$ of $Z(\tau)$, depending on
 $c_{4}$ s.t. $Vol(B_{i})\leq c_{4}$. Next we get by \eqref{multip}-\eqref{Greenform} that
\begin{equation}\label{R<0}
\int_{M}{R\varphi^{\frac{(N+2)q_{i}}{2}}dv}+\frac{n-1}{n}\int_{B_{k_{0}}\setminus B_{i}}{\tau^{2}\varphi^{N+\frac{(N+2)q_{i}}{2}-2}dv}\leq 2\left(\|\sigma\|_{L^{q_{i}}}^{2}+r_{i}^{2}\right)\left(\int_{M}{\varphi^{\frac{(N+2)q_{i}}{2}}dv}\right)^{\frac{q_{i}-2}{q_{i}}}.
\end{equation}
Set $\tau_{i}=\inf_{M\setminus B_{i}}|\tau|>0$. Since $\varphi^{\frac{(N+2)q_{i}}{4}}\geq \frac{1}{k_{0}}\|\varphi^{\frac{(N+2)q_{i}}{4}}\|_{L^{2}}$ on $B_{k_{0}}$ and since $Vol\left(B_{k_{0}}\setminus B_{i}\right)\geq c_{4}$, it follows from (\ref{R<0}) that
$$-\|R\|_{L^{\infty}}\|\varphi^{\frac{(N+2)q_{i}}{4}}\|_{L^{2}}^{2}+\frac{n-1}{n}c_{4}\tau_{i}^{2}\left(\frac{\|\varphi^{\frac{(N+2)q_{i}}{4}}\|_{L^{2}}}{k_{0}}\right)^{2\left(\frac{(q_{i}+2)(N+2)-8}{q_{i}(N+2)}\right)}\leq 2\left(\|\sigma\|_{L^{q_{i}}}^{2}+r_{i}^{2}\right)\left(\int_{M}{\varphi^{\frac{(N+2)q_{i}}{2}}dv}\right)^{\frac{q_{i}-2}{q_{i}}}.$$
Since $\frac{q_{i}-2}{q_{i}}<1<\frac{(q_{i}+2)(N+2)-8}{q_{i}(N+2)}$ for all $i\in\mathbb{N}$, we get from the previous inequality that $\|\varphi^\frac{(N+2)q_{i}}{4}\|_{L^{2}}$ is bounded by $c_{5}=c_{5}(n,g,\tau,\sigma,r_{i},c_{4},k_{0},q_{i})$.
\end{itemize}
In both cases, we have showed that 
$\|\varphi^{\frac{(N+2)q_{i}}{4}}\|_{L^{2}}\leq c_{5}$ and hence by \eqref{L2} that
\begin{equation}\label{key_for_LW}
\|\varphi^{\frac{(N+2)q_{i}}{4}}\|_{L^{N}}\leq c_{6}(c_{5},c_{3}).
\end{equation}
Now by the Sobolev embedding theorem, from vector equation, there exists $c_{7}=c_{7}(M,g)$ s.t.
\begin{equation}\label{induction}
\begin{aligned}
 \|LW\|_{L^{\frac{nq_{i}(N+2)}{\left(4n-(N+2)q_{i}\right)^{+}}}}&\leq c_{7}\|\varphi^{N}d\tau\|_{L^{\frac{(N+2)q_{i}}{4}}}\\
&\leq c_{7}\|d\tau\|_{\infty}\|\varphi^{N}\|_{L^{\frac{(N+2)q_{i}}{4}}}~~\left(\mbox{since $\tau\in C^{1}$}\right)\\
&\leq c_{8}(c_{7},\tau)\|\varphi^{\frac{(N+2)q_{i}}{4}}\|_{L^{N}}^{\frac{4N}{(N+2)q_{i}}}\\
&\leq c_{9}(c_{8},c_{6})\qquad\qquad\qquad\mbox{(by (\ref{key_for_LW}))}.
\end{aligned}
\end{equation}
Here $\left(4n-(N+2)q_{i}\right)^{+}=\max\{4n-(N+2)q_{i},0\}$ and $L^{\frac{nq_{i}(N+2)}{\left(4n-(N+2)q_{i}\right)^{+}}}$ is understood to be $L^{\infty}$ if $4n\leq (N+2)q_{i}$. Since $q_{i+1}<\frac{nq_{i}(N+2)}{\left(4n-(N+2)q_{i}\right)^{+}}$, it follows from (\ref{induction}) that $\|LW\|_{L^{q_{i+1}}}\leq r_{i+1}(n,g,\tau,\sigma,q_{i},r_{i})$ as claimed.\\

Finally, note that $\frac{N+2}{4}>1$, we can then take $i_{0}$ large enough depending only on $n$ s.t. $q_{i_{0}}\leq \left[\frac{4n}{N+2}\right]+1$. Thus, applying inductively (\ref{induction}) for $i\leq i_{0}$, provided $\|LW\|_{L^{q_{0}}}=\|LW\|_{L^{2}}\leq l$, we obtain that $|LW|$ is uniformly bounded in $L^{\infty}$ by $c=c(n,g,\tau,\sigma,l)>0$, which completes our proof.
\end{proof}

We are now ready to prove the second main result of this paper.

\begin{Theorem}[Near-CMC]\label{new-nearCMC}
Assume that $\tau\in L^{\infty}$, $\xi\in L^{\infty}$, $g\in W^{2,p}$ $(p>n)$, $(M,g)$ has no conformal Killing vector field, and $\sigma\nequiv 0$ if $\mathcal{Y}_{g}\geq 0$.
Assume further that $\|\frac{\xi}{\tau}\|_{L^{n}}$ is small enough, then the system equations \eqref{Modified CE} admits a solution $(\varphi,W)$.
\end{Theorem}

\begin{proof}
Recall that $T$, defined in Section 2 (where $d\tau$ is replaced by $\xi$ in the vector equation), is a continuous compact map and $T(\varphi)>0$ for all
$\varphi\in L^{\infty}$. As explained in Remark \ref{assumptionR}, there exists a constant $\kappa_{1}=\kappa_{1}(g,\tau)$ s.t. 

\begin{equation}\label{kappa1}
RT(\varphi)^{N+2}+\frac{n-2}{n}\tau^{2}T(\varphi)^{2N}\geq \kappa_{1},~~~\forall \varphi\in L^{\infty}.
\end{equation}

Set $\kappa=\max\left\{|\kappa_{1}|, \int_{M}{|\sigma|^{2}dv}\right\}$. Let $S$ be given by
\begin{equation}\label{define:S}
S(\varphi)=\left\{\begin{array}{ll}
\min\{T(\varphi),a\}&\textrm{if $\|LW_{\varphi}\|_{L^{2}}\leq \sqrt{\kappa}$,}\\
0&\textrm{otherwise,}
\end{array} \right.
\end{equation}
and set $\mathscr{C}=\left\{\varphi\in C^{0}:0\leq\varphi\leq a\right\}$, where $a$ will be determined later. \\
\\
Since $T$ is a continuous compact map from $C^{0}$ to $C^{0}_{+}$ and since by definition $0\leq S(\varphi)\leq a$ for all $\varphi$,
$S$ maps $\mathscr{C}$ into itself and $S(\mathscr{C})$ is precompact. Assume for the moment that the half-continuity of $S$ is proven.
By Corollary \ref{Fixed point theorem}, $S$ has a fixed point $\varphi_{0}$. Note that $\varphi_{0}$ is not zero otherwise
$0=\varphi_{0}=S(\varphi_{0})$, hence $\|LW_{\varphi_{0}}\|_{L^{2}}=0\leq \sqrt{\kappa}$. We get from the definition of $S$ that
$S(\varphi_{0})=\min\left\{T(\varphi_{0}),a\right\}>0$ which is a contradiction with $S(\varphi_{0})=0$. Since $\varphi_{0}\nequiv 0$,
so is $S(\varphi_{0})$, the definition of $S$ implies that $\|LW_{\varphi_{0}}\|_{L^{2}}\leq \sqrt{\kappa}$ and

\begin{equation}\label{equal}
\varphi_{0}=\min\{T(\varphi_{0}),a\}\leq T(\varphi_{0}).
\end{equation}

Set 
$$
K=\left\{\varphi:~\|LW_{\varphi}\|_{L^{2}}\leq \sqrt{\kappa}~~\mbox{and}~~\varphi\leq T(\varphi)\right\}.
$$
Arguing as in the proof of Theorem \ref{Estimate1 L2}, we obtain that if any $\varphi\in K$ satisfies $\|LW_{\varphi}\|_{L^{q_{i}}}\leq r_{i}$, then

\begin{equation}\label{bounded_near_CMC}
\left\|T(\varphi)^{\frac{(N+2)q_{i}}{4}}\right\|_{L^{N}}\leq \tilde{r}_{i}(n,g,\tau,\sigma,r_{i},q_{i}),
\end{equation}
where $q_{i}=2\left(\frac{N+2}{4}\right)^{i}$ for all $i\in\mathbb{N}$. Therefore, by the Sobolev embedding theorem, we have from the vector equation that

\begin{equation}\label{induction2}
\begin{aligned}
 \|LW_{\varphi}\|_{L^{\frac{nq_{i}(N+2)}{\left(4n-(N+2)q_{i}\right)^{+}}}}&\leq r(M,g)\|\varphi^{N}\xi\|_{L^{\frac{(N+2)q_{i}}{4}}}\\
&\leq r\|\xi\|_{\infty}\|\varphi^{N}\|_{L^{\frac{(N+2)q_{i}}{4}}}~~\qquad\quad~~\left(\mbox{since $\xi\in L^{\infty}$}\right)\\
&\leq r\|\xi\|_{\infty}\|\varphi^{\frac{(N+2)q_{i}}{4}}\|_{L^{N}}^{\frac{4N}{(N+2)q_{i}}}\\
&\leq r\|\xi\|_{\infty}\left\|T(\varphi)^{\frac{(N+2)q_{i}}{4}}\right\|_{L^{N}}^{\frac{4N}{(N+2)q_{i}}}\qquad  (\mbox{by $\varphi\leq T(\varphi)$})\\
&\leq r_{i+1}(\xi,r,\tilde{r}_{i})\quad\qquad\qquad\qquad\mbox{(by (\ref{bounded_near_CMC}))},
\end{aligned}
\end{equation}
where $\left(4n-(N+2)q_{i}\right)^{+}=\max\{4n-(N+2)q_{i},0\}$ and $L^{\frac{nq_{i}(N+2)}{\left(4n-(N+2)q_{i}\right)^{+}}}$ is understood to be $L^{\infty}$ if
$4n\leq (N+2)q_{i}$. Similarly to the proof of Theorem \ref{Estimate1 L2}, we obtain inductively from \eqref{induction2} that for all $\varphi\in K$, there exists
a constant $C=C(n,g,\tau,\xi,\kappa)>0$ s.t.
$$\|LW_{\varphi}\|_{L^{\infty}}\leq C,$$
and hence by Lemma \ref{maxpriw} the set $T(K)$ is bounded by $\max\psi_{C}$, where $\psi_{C}$ is the unique positive solution to the Lichnerowicz equation
\eqref{Lichnerowicz} associated to $w=\|\sigma\|_{L^{\infty}}+C$. Thus, taking $a=\max\psi_{C}+1$, since $\varphi_{0}\in K$, we also obtain from \eqref{equal}
that $\varphi_{0}=T(\varphi_{0})$, which proves the theorem.\\

We now prove the half-continuity of $S$. Since $T$ is continuous, so is $S$ at $\varphi$ satisfying $\|LW_{\varphi}\|_{L^{2}}\ne\sqrt{\kappa}$. For $\varphi$ s.t.
$\|LW_{\varphi}\|_{L^{2}}=\sqrt{\kappa}$, multiplying the Lichnerowicz equation by $T(\varphi)^{N+1}$ and integrating over $M$, we have 
 \begin{align*}
 \frac{4(n-1)\left(N+1\right)}{(n-2)\left(\frac{N}{2}+1\right)^{2}}\int_{M}{|\nabla T(\varphi)^{\frac{N+2}{2}}|^{2} dv}+\int_{M}{RT(\varphi)^{N+2} dv}+\frac{n-1}{n}\int_{M}{\tau^{2}T(\varphi)^{2N}dv}&=\int_{M}{|\sigma+LW_{\varphi}|^{2}dv}\\
 &=\int_{M}{|\sigma|^{2}dv}+\int_{M}{|LW_{\varphi}|^{2}dv}\\
 &=\int_{M}{|\sigma|^{2}dv}+\kappa.
 \end{align*}
 Therefore,
 \begin{equation}\label{nearCMC1}
 \begin{aligned}
  \int_{M}{\tau^{2} T(\varphi)^{2N}dv}&\leq n\left(\int_{M}{|\sigma|^{2}dv}+\kappa-\left(\int_{M}{RT(\varphi)^{N+2} dv}+\frac{n-2}{n}\int_{M}{\tau^{2}T(\varphi)^{2N}dv}\right)\right) \\
  &\leq n\left(\int_{M}{|\sigma|^{2}dv}+\kappa+|\kappa_{1}|\right)~~~~~~\mbox{(by (\ref{kappa1}))}\\
  &\leq 3n\kappa.
 \end{aligned}
 \end{equation}
 On the other hand, we get from the vector equation that
 \begin{equation}\label{nearCMC2}
 \begin{aligned}
 \kappa=\int_{M}{|LW_{\varphi}|^{2}dv}&\leq C_{5}(g)\left\|W_{\varphi}\right\|^{2}_{W^{2,\frac{2n}{n+2}}}~~~~~~~~~~~(\mbox{by Sobolev imbedding})\\
 &\leq C_{6}(g,C_{5}) \|L^{*}LW_{\varphi}\|_{L^{\frac{2n}{n+2}}}^{2}\\
 &\leq C_{7}(C_{6})\left(\int_{M}{|\xi|^{\frac{2n}{n+2}}\varphi^{\frac{2nN}{n+2}}dv}\right)^{\frac{n+2}{n}}\\
 &\leq C_{7}\left\|\frac{\xi}{\tau}\right\|_{L^{n}}^{2}\int_{M}{\tau^{2}\varphi^{2N}dv}~~~(\mbox{by H\"{o}lder inequality})
 \end{aligned}
 \end{equation}
 
By \eqref{nearCMC1} and \eqref{nearCMC2}, we obtain that $$\int_{M}{\tau^{2} T(\varphi)^{2N}dv}\leq 3nC_{7}\left\|\frac{\xi}{\tau}\right\|_{L^{n}}^{2}\int_{M}{\tau^{2}\varphi^{2N}dv}.$$ 
If $\left\|\frac{\xi}{\tau}\right\|_{L^{n}}$ is small enough s.t. $3nC_{7}\left\|\frac{\xi}{\tau}\right\|_{L^{n}}^{2}<1$, it follows from the previous inequality that there exists
$m\in M$ s.t.  $0<T(\varphi)(m)<\varphi(m)$ (note that $T(\varphi)\in C^{0}_{+}$). Therefore, since $T$ is continuous, there exists $\delta=\delta(\varphi)>0$ small enough s.t.
$$0<T(\psi)(m)<\psi(m),~~\forall\psi\in B(\varphi,\delta)\cap \mathscr{C},$$
and hence from the fact that

\begin{equation*}
-\bigl(S(\psi)(m)-\psi(m)\bigr)=\left\{\begin{array}{ll}
-\bigl(\min\{T(\psi)(m),a\}-\psi(m)\bigr)&\textrm{if $\|LW_{\psi}\|_{L^{2}}\leq \sqrt{\kappa},$}\\
\psi(m)&\textrm{otherwise,}
\end{array} \right.
\end{equation*}
we conclude that 
\begin{equation}\label{obtain_halfcontinuous}
-\bigl(S(\psi)(m)-\psi(m)\bigr)>0
\end{equation}
for all $\psi\in B(\varphi,\delta)\cap \mathscr{C}$.\\

Now let $p~:~C^{0}\longrightarrow \mathbb{R}$ be defined by $p(f)=-f(m)$ for all $f\in C^{0}$. It is obvious that $p\in \left(C^{0}\right)^{*}$. Moreover, Inequality
\eqref{obtain_halfcontinuous} tells us that $p\left(S(\psi)-\psi\right)>0$ for all $\psi\in B(\varphi,\delta)\cap \mathscr{C}$, and then by definition $S$ is
half-continuous at $\varphi$ as claimed. The proof is completed.
\end{proof}

Our next existence result deals with the far-from-CMC case. It makes some progresses compared with the statements of Holst-Nagy-Tsogtgerel \cite{HNT} and Maxwell \cite{Maxwell}
(see Proposition \ref{farCMCofMawell}), where the smallness assumption on $\sigma$ is in $L^{\infty}$. Here our assumption is on the $L^2$-norm of $\sigma$.

\begin{Theorem}[\textbf{Far-from-CMC}]\label{far-CMC}
Let data be given on $M$ as specified in \eqref{condintial}. Assume that $\mathcal{Y}(g)>0$, $(M,g)$ has no conformal Killing vector field and $\sigma\nequiv 0$.
If $\|\sigma\|_{L^{2}}$ is small enough (depending only on $g$ and $\tau$), then the system \eqref{CE} has a solution $(\varphi,W)$.
\end{Theorem}

\begin{proof} Regarding Remark \ref{assumptionR}, we may assume that $R>0$. We define

\begin{equation}\label{definefar-CMC}
S(\varphi)=\left\{\begin{array}{ll}
\min\{T(\varphi),a\} &\textrm{if $\frac{4(N+1)}{(N+2)^{2}}\mathcal{Y}_{g}\left(\int_{M}{\varphi^{\frac{N(N+2)}{2}}dv}\right)^{\frac{2}{N}}\leq 2\int_{M}{|\sigma|^{2}dv}$}\\
0&\textrm{otherwise,}
\end{array} \right.
\end{equation}
where $a$ is to be determined later. Let $$\mathscr{C}=\left\{\varphi\in C^{0}(M):~ \|\varphi\|_{\infty}\leq a\right\}.$$

Similarly to the previous proof, $S$ maps $\mathscr{C}$ into itself and $S(\mathscr{C})$ is precompact since $T$ is a compact map from $C^{0}$ into $C^{0}_{+}$.
Assume that the half-continuity of $S$ is proven. Then Corollary \ref{Fixed point theorem} implies that $S$ admits a fixed point $\varphi_{0}$. Note that $\varphi_{0}$
is not zero. Indeed, if $0=\varphi_{0}=S(\varphi_{0})$, it follows that
$\frac{4(N+1)}{(N+2)^{2}}\mathcal{Y}_{g}\left(\int_{M}{\varphi_{0}^{\frac{N(N+2)}{2}}dv}\right)^{\frac{2}{N}}=0\leq 2\int_{M}{|\sigma|^{2}dv}$, and hence from the definition of
$S$ we get that $S(\varphi_{0})=\min\left\{T(\varphi_{0}),a\right\}>0$ which is a contradiction with $S(\varphi_{0})=0$. Since $\varphi_{0}\nequiv 0$, so is $S(\varphi_{0})$,
and the definition of $S$ implies that
$$
\frac{4(N+1)}{(N+2)^{2}}\mathcal{Y}_{g}\left(\int_{M}{\varphi_{0}^{\frac{N(N+2)}{2}}dv}\right)^{\frac{2}{N}}\leq 2\int_{M}{|\sigma|^{2}dv}
~~~\mbox{and}~~~
\varphi_{0}=S(\varphi_{0})=\min\{T(\varphi_{0}),a\}\leq  T(\varphi_{0}).
$$ 
On the other hand, the first condition on $\varphi_{0}$ and the smallness assumption on $\|\sigma\|_{L^{2}}$ implies
$$\int_{M}{|LW_{\varphi_{0}}|^{2}dv}\leq \int_{M}{|\sigma|^{2}dv}.$$
Indeed,
\begin{equation}\label{LWandSigma}
\begin{aligned}
\int_{M}{|LW_{\varphi_{0}}|^{2}dv}&\leq C(g)\|\varphi_{0}^{N} d\tau\|^{2}_{L^{\frac{2n}{n+2}}}~~\mbox{(by Sobolev imbedding theorem)}\\ 
&\leq C\|d\tau\|_{L^{p}}^{2}\left(\int_{M}{\varphi_{0}^{\frac{2nNp}{(n+2)p-2n}}dv}\right)^{\frac{(n+2)p-2n}{np}}~~\mbox{(by H\"{o}lder inequality)}\\
     &\leq C\|d\tau\|_{L^{p}}^{2}\left(\int_{M}{\varphi_{0}^{\frac{N(N+2)}{2}}dv}\right)^{\frac{4}{N+2}}~~\mbox{(by H\"{o}lder inequality and $p>n$)}\\ &\leq C\|d\tau\|_{L^{p}}^{2}\left(\frac{(N+2)^{2}}{2(N+1)\mathcal{Y}_{g}}\right)^{\frac{2N}{N+2}}\|\sigma\|_{L^{2}}^{\frac{2(N-2)}{N+2}}\int_{M}{|\sigma|^{2}dv}~~\mbox{(by the first condition on $\varphi_{0}$)}\\
     &\leq \int_{M}{|\sigma|^{2}dv},
\end{aligned}
\end{equation}
where the last inequality holds provided $\|\sigma\|_{L^{2}}$ is small enough so that
$C\|d\tau\|_{L^{p}}^{2}\left(\frac{(N+2)^{2}}{2(N+1)\mathcal{Y}_{g}}\right)^{\frac{2N}{N+2}}\|\sigma\|_{L^{2}}^{\frac{2(N-2)}{N+2}}\leq 1$.
Setting
$$K=\left\{T(\varphi):~~\|LW_{\varphi}\|_{L^{2}}\leq \|\sigma\|_{L^{2}}~~\mbox{and}~~\varphi\leq T(\varphi)\right\},$$
similarly to the proof of Theorem \ref{new-nearCMC}, we then obtain that $T(K)$
is uniformly bounded in $L^{\infty}$ by $C=C(g,\tau,\sigma)$. Thus, taking $a\geq C$, since $\varphi_{0}\in K$, we obtain from the second condition on $\varphi_{0}$ that $\varphi_{0}=T(\varphi_{0})$, which completes our proof. \\
 
Now we prove the half-continuity of $S$ on $\mathscr{C}$. Since $T$ is continuous, so is $S$ at $\varphi$ satisfying 
$$\frac{4(N+1)}{(N+2)^{2}}\mathcal{Y}_{g}\left(\int_{M}{\varphi^{\frac{N(N+2)}{2}}dv}\right)^{\frac{2}{N}}\ne 2\int_{M}{|\sigma|^{2}dv}.$$

For the remaining $\varphi$, i.e. when $\frac{4(N+1)}{(N+2)^{2}}\mathcal{Y}_{g}\left(\int_{M}{\varphi^{\frac{N(N+2)}{2}}dv}\right)^{\frac{2}{N}}= 2\int_{M}{|\sigma|^{2}dv}$,
first note that, arguing as to get \eqref{LWandSigma}, we have
\begin{equation}\label{LWandSigma2}
\int_{M}{|LW_{\varphi}|^{2}dv}\leq \int_{M}{|\sigma|^{2}dv}.
\end{equation}

Next we prove that there exists $m\in M$ s.t. $\varphi(m)> T(\varphi)(m)$. We argue by contradiction. Assume that it is not true, then
\begin{equation}\label{farkey}
\begin{aligned}
\frac{4(N+1)}{(N+2)^{2}}\mathcal{Y}_{g}\left(\int_{M}{T(\varphi)^{\frac{N(N+2)}{2}}dv}\right)^{\frac{2}{N}}&\geq\frac{4(N+1)}{(N+2)^{2}}\mathcal{Y}_{g}\left(\int_{M}{\varphi^{\frac{N(N+2)}{2}}dv}\right)^{\frac{2}{N}}= 2\int_{M}{|\sigma|^{2}dv}\\
&\geq \int_{M}{|\sigma|^{2}dv}+\int_{M}{|LW|^{2}dv}~~~~\left(\mbox{by (\ref{LWandSigma2})}\right).
\end{aligned}
\end{equation}
On the other hand, multiplying the Lichnerowicz equation by $T(\varphi)^{N+1}$ and integrating over $M$, we obtain 
\begin{equation}\label{contradiction}
\frac{16(n-1)(N+1)}{(n-2)(N+2)^{2}}\int_{M}{|\nabla T(\varphi)^{\frac{N+2}{2}}|^{2} dv}+\int_{M}{RT(\varphi)^{N+2} dv}+\frac{n-1}{n}\int_{M}{\tau^{2}T(\varphi)^{2N}dv}
=\int_{M}{|\sigma|^{2}dv}+\int_{M}{|LW_{\varphi}|^{2}dv}.
\end{equation}
Since
\begin{align*}
\frac{16(n-1)(N+1)}{(n-2)(N+2)^{2}}\int_{M}{|\nabla T(\varphi)^{\frac{N+2}{2}}|^{2} dv}&+\int_{M}{RT(\varphi)^{N+2} dv}\\
&\geq \frac{4(N+1)}{(N+2)^{2}}\left(\frac{4(n-1)}{n-2}\int_{M}{|\nabla T(\varphi)^{\frac{N+2}{2}}|^{2} dv}+\int_{M}{RT(\varphi)^{N+2} dv}\right)\quad(\mbox{since $R>0$})\\
&\geq \frac{4(N+1)}{(N+2)^{2}}\mathcal{Y}_{g}\left(\int_{M}{T(\varphi)^{\frac{N(N+2)}{2}}dv}\right)^{\frac{2}{N}}~~~\mbox{(by the definition of $\mathcal{Y}_{g}$)}\\
&\geq \int_{M}{|\sigma|^{2}dv}+\int_{M}{|LW|^{2}dv},~~~~~~~~~~~~~~~~~~~~(\mbox{by (\ref{farkey})})
\end{align*}  
it follows from \eqref{contradiction} that $\int_{M}{\tau^{2}T(\varphi)^{2N}dv}\leq 0$, which is a contradiction.\\

Now let $m\in M$ s.t. $0<T(\varphi)(m)<\varphi(m)$ (note that $T(\varphi)\in C_{+}^{0}$). By the continuity of $T$, we obtain
that there exists $\delta=\delta(\varphi)$ s.t. for all $\psi\in B(\varphi,\delta)\cap \mathscr{C}$,
$$0<T(\psi)(m)<\psi(m),$$
and hence from the fact that
\begin{equation*}
-\bigl(S(\psi)(m)-\psi(m)\bigr)=\left\{\begin{array}{ll}
-\bigl(\min\{T(\psi)(m),a\}-\psi(m)\bigr)&\textrm{if $\frac{4(N+1)}{(N+2)^{2}}\mathcal{Y}_{g}\left(\int_{M}{\psi^{\frac{N(N+2)}{2}}dv}\right)^{\frac{2}{N}}\leq 2\int_{M}{|\sigma|^{2}dv}$}\\
\psi(m)&\textrm{otherwise,}
\end{array} \right.
\end{equation*}
we conclude that $-\bigl(S(\psi)(m)-\psi(m)\bigr)>0,~\forall \psi\in B(\varphi,\delta)\cap \mathscr{C}$.

Hence, by the definition of  half-continuity applied with $p(f)=-f(m)$ for all $f\in C^{0}$, we obtain that $S$ is half-continuous at $\varphi$. The proof is completed.
\end{proof}

\begin{Remark}\label{remarkfarcmc}
From the proof above, a more precise assumption for Theorem \ref{far-CMC} is that $\|d\tau\|_{L^{p}}^{2}\|\sigma\|_{L^{2}}^{\frac{2(N-2)}{N+2}}$ is small enough, only depending on $(M,g)$. 
\end{Remark}

\subsection{A Sufficient Condition to the Existence of Solutions}
We note that the main ingredient to prove the half-continuity of $S$ in the two proofs above is the existence of $m\in M$ s.t. $T(\varphi)(m)<\varphi(m)$.
This leads us to propose a sufficient condition for the existence of a solution to \eqref{CE}, which is much weaker than the concept of a global supersolution
(see \cite{HNT} or \cite{Maxwell}). We will begin with the notion of a local supersolution.

\begin{Definition}
Let data be given on $M$ as specified in \eqref{condintial} and assume that \eqref{condinitial2} holds. We call $\psi\in L^{\infty}_{+}$ a local supersolution to
\eqref{CE} if for every positive function $\varphi$ satisfying $\varphi\leq \psi$ and $\varphi=\psi$ somewhere, then there exists $m\in M$ such that $T(\varphi)(m)\leq \varphi(m)$.
\end{Definition}
\noindent Recall that $\psi\in L^{\infty}_{+}$ is called a global supersolution to \eqref{CE} if for all $m\in M$,
$$\sup_{\substack{\varphi\leq \psi,\\ \varphi\in L^{\infty}_{+}}}{T(\varphi)(m)}\leq \psi(m).$$
It follows immediately that
\begin{Proposition}\label{local-global}
A global supersolution is a local supersolution.
\end{Proposition}
\begin{proof} Assume that $\psi$ is a global supersolution to \eqref{CE}. Let $\varphi$ be an arbitrary positive function satisfying $\varphi\leq \psi$ and $\varphi=\psi$ somewhere.
Taking $m\in M$ s.t. $\varphi(m)=\psi(m)$, by definition of a global supersolution, it is clear that 
$$T(\varphi)(m)\leq \psi(m)= \varphi(m),$$
and hence $\psi$ is a local supersolution.
\end{proof}
\begin{Theorem}\label{localsuper}
 Let data be given on $M$ as specified in \eqref{condintial} and assume that \eqref{condinitial2} holds. Assume that $\psi\in L^{\infty}_{+}$ is a local supersolution to \eqref{CE},
 then \eqref{CE} admits a solution.
\end{Theorem}
\begin{proof}  Let $\mathscr{C}$ be given by
$$\mathscr{C}=\left\{\varphi\in C^{0}:0\leq\varphi\leq b\right\},$$
with $b$ large enough s.t.
$$\sup_{\varphi\leq \psi}{\|T(\varphi)\|_{\infty}}<b.$$
Here recall that from the vector equation, the set  $\left\{LW_{\varphi}:\varphi\leq \psi\right \}$ is uniformly bounded in $L^{\infty}$ by $b_{1}=b_{1}(M,g,\psi,\tau)$. Then, by
Lemma \ref{maxpriw}, $\{T(\varphi):~\varphi\leq \psi\}$ is uniformly bounded (in $L^{\infty}$) by $\max{\psi_{0}}$, where $\psi_{0}$ is the unique solution to \eqref{Lichnerowicz}
associated to $w= b_{1}+\|\sigma\|_{\infty}$, and hence $b$ is well-defined.\\

We define
\begin{equation}\label{define:S1}
S(\varphi)=\left\{\begin{array}{ll}
T(\varphi)&\textrm{if $\varphi\leq \psi$}\\
0&\textrm{otherwise.}
\end{array} \right.
\end{equation}
By Proposition \ref{compacityofT}, $T$ is a compact map from $C^{0}$ into $C^{0}_{+}$. Then $S$ maps $\mathscr{C}$ into itself and $S(\mathscr{C})$ is precompact. Assume for the moment that
the half-continuity of $S$ is proven. By Corollary \ref{Fixed point theorem}, $S$ has a fixed point $\varphi_{0}$. We claim that  $\varphi_{0}\nequiv 0$. Indeed, if is not true, then
$0=\varphi_{0}=S(\varphi_{0})$, hence $\varphi_{0}=0\leq \psi$. We get from the definition of $S$ that $S(\varphi_{0})=T(\varphi_{0})>0$ which is a contradiction with $S(\varphi_{0})=0$.
Since $\varphi_{0}\nequiv 0$, so is $S(\varphi_{0})$, and the definition of $S$ implies that  $\varphi_{0}=S(\varphi_{0})=T(\varphi_{0})$.\\

Now we prove the half-continuity of $S$ on $\mathscr{C}$. Since $T$ is continuous, so is $S$ at $\varphi$ satisfying $\varphi<\psi$ everywhere or $\varphi>\psi$ somewhere. The only
remaining work is to show that $S$ is half-continuous at $\varphi$ s.t. $\varphi\leq \psi$ and $\varphi=\psi$ somewhere.\\

For such a $\varphi$, assume that there exists $m_{0}\in M$ s.t.
$$T(\varphi)(m_{0})<\varphi(m_{0})$$
By the continuity of $T$, we can choose $\delta=\delta(\varphi)>0$ s.t. for all $\eta\in B(\varphi,\delta)\cap \mathscr{C}$,
$$T(\eta)(m_{0})<\eta(m_{0}),$$
and hence from the fact that
\begin{equation*}
-\bigl(S(\eta)(m_{0})-\eta(m_{0})\bigr)=\left\{\begin{array}{ll}
-\bigl(T(\eta)(m_{0})-\eta(m_{0})\bigr)&\textrm{if $\eta\leq \psi$}\\
\eta(m_{0})&\textrm{otherwise,}
\end{array} \right.
\end{equation*}
we obtain that $-\bigl(S(\eta)(m_{0})-\eta(m_{0})\bigr)>0,~\forall \eta\in B(\varphi,\delta)\cap \mathscr{C}$. Now, by the definition of half-continuity applied with $p(f)=-f(m_{0})$
for all $f\in C^{0}$, we conclude that $S$ is half-continuous at $\varphi$.\\

It remains to study the case when $\varphi\leq T(\varphi)$. Since $\psi$ is a local supersolution, there exists $m$ s.t. $T(\varphi)(m)\leq\varphi(m)$ and since $\varphi\leq T(\varphi)$,
we have $T(\varphi)(m)=\varphi(m)$. Because the case $T(\varphi)\equiv \varphi$ is trivial $\bigl[$$(\varphi,W_{\varphi})$ is then a solution to \eqref{CE}$\bigr]$, we can assume that
there exists $q\in M$ s.t. $T(\varphi)(q)> \varphi(q)$. Let $A,B>0$ satisfying
\begin{equation}\label{sfcd1}
A\varphi(m)-B\varphi(q)>0.
\end{equation}
Note that since $\varphi(m)=T(\varphi)(m)>0$ ($T(\varphi)\in C_{+}^{0}$), such $A,B$ exist. On the other hand, by the assumptions on $q$ and $m$,
\begin{equation}\label{sfcd2}
-A\bigl(T(\varphi)(m)-\varphi(m)\bigr)+B\bigl(T(\varphi)(q)-\varphi(q)\bigr)=-A.0+B\bigl(T(\varphi)(q)-\varphi(q)\bigr)>0.
\end{equation}
By (\ref{sfcd1}), (\ref{sfcd2}) and the continuity of $T$, there exists $\delta_{1}=\delta_{1}(\varphi)>0$ small enough s.t. for all $\eta\in B(\varphi,\delta_{1})\cap \mathscr{C}$
$$A\eta(m)-B\eta(q)>0$$
and
$$-A\left(T(\eta)(m)-\eta(m)\right)+B\left(T(\eta)(q)-\eta(q)\right)>0.$$
Therefore, by the fact that
\begin{equation*}
-A\bigl(S(\eta)(m)-\eta(m)\bigr)+B\bigl(S(\eta)(q)-\eta(q)\bigr)=\left\{\begin{array}{ll}
-A\left(T(\eta)(m)-\eta(m)\right)+B\left(T(\eta)(q)-\eta(q)\right)&\textrm{if $\eta\leq \psi$}\\
A\eta(m)-B\eta(q)&\textrm{otherwise,}
\end{array} \right.
\end{equation*}
we obtain that $-A\bigl(S(\eta)(m)-\eta(m)\bigr)+B\bigl(S(\eta)(q)-\eta(q)\bigr)>0$ for all $\eta\in B(\varphi,\delta_{1})\cap\mathscr{C}$.
Now, by the definition of half-continuity applied with $p(f)=-Af(m)+Bf(q)$ for all $f\in C^{0}$, we can conclude that $S$ is half-continuous at
$\varphi$. The proof is completed.
\end{proof}

A direct consequence of Theorem \ref{localsuper} is the following:
\begin{Corollary}
For every $\varphi\in L^{\infty}$ large enough, if $T(\varphi)\leq \varphi$ somewhere,  then \eqref{CE} admits a solution. 
\end{Corollary}

\renewcommand{\bibname}{References}


{\small Laboratoire de Mathématiques et Physique Théorique Université de Tours,\\
UFR Sciences et Techniques,\\
Parc de Grandmont,\\
37200 Tours - FRANCE\\
\textit{E-mail:  The-Cang.Nguyen@lmpt.univ-tours.fr}}


\begin{thebibliography}{99}
\addcontentsline{toc}{section}{References}
\bibitem{Au} {T. Aubin},
  \textit	 {Some nonlinear problems in {R}iemannian geometry},
  {Springer Monographs in Mathematics},
  {Springer-Verlag},
  {Berlin},
  {1998}.


\bibitem{PB} {P. Bich},
  \textit{Some fixed point theorems for discontinuous mappings},
  {2006},
  {Université Panthéon-Sorbonne (Paris 1)},
  {Cahiers de la Maison des Sciences Economiques}.
  

\bibitem{BI}
    {R. Bartnik and J. Isenberg},
     \textit {The constraint equations},
 {The {E}instein equations and the large scale behavior of
              gravitational fields},
  {1-38},
 {Birkh\"auser},
  {Basel},
      {2004}.


\bibitem{KCC} {Chang Kung-Ching},
     \textit {Methods in nonlinear analysis},
   {Springer Monographs in Mathematics},
  {Springer-Verlag, Berlin},
  {2005},
  {x+439}.
 

\bibitem{CB}
     {Y. Four{\`e}s-Bruhat},
    \textit{Th\'eor\`eme d'existence pour certains syst\`emes
              d'\'equations aux d\'eriv\'ees partielles non lin\'eaires},
   {Acta Math.},
    {88},
    {1952},
    {141-225}.


\bibitem {CBG69}
   {Y. Choquet-Bruhat and R. Geroch},
     \textit {Global aspects of the {C}auchy problem in general relativity},
  {Comm. Math. Phys.},
  {14},
  {1969},
  {329-335}.


\bibitem{CBY80}
 {Y. Choquet-Bruhat and Jr. York},
  \textit	 {The {C}auchy problem},
  {General relativity and gravitation, {V}ol. 1},
  {99--172},
  {Plenum},
  {New York},
  {1980}.


\bibitem{DGH}
  {M. Dahl and R. Gicquaud and E. Humbert},
     \textit {A limit equation associated to the solvability of the vacuum {E}instein constraint equations by using the conformal method},  {Duke Math. J.},
 {Duke Mathematical Journal},
 {161},
 {2012},
 {14},
 {2669-2697}.
 
 
\bibitem{DT}
    {D. Gilbarg and N.S Trudinger},
     \textit {Elliptic partial differential equations of second order},
    {Classics in Mathematics},
    {Reprint of the 1998 edition},
    {Springer-Verlag},
    {Berlin},
    {2001},
   {xiv+517}.


\bibitem{G-NQA}
    {R. Gicquaud and Q.A. Ng\^{o}},
     \textit {A new point of view on the solutions to the {E}instein constraint equations with arbitrary mean curvature and small {TT}-tensor},
  {Class. Quantum Grav.},
  {31},
{2014},
  {19},
  {195014 (20pp)}.


\bibitem{HNT}
{M. Holst and G. Nagy and G. Tsogtgerel},
     \textit {Rough solutions of the {E}instein constraints on closed manifolds without near-{CMC} conditions},
{Comm. Math. Phys.},
{288},
{2009},
{2},
{547-613}.


\bibitem{IM}
 {J. Isenberg and V. Moncrief},
  \textit	 {Some results on nonconstant mean curvature solutions of the {E}instein constraint equations},
 {Physics on manifolds ({P}aris, 1992)},
 {Math. Phys. Stud.},
 {15},
 {295-302},
 {Kluwer Acad. Publ.},
 {Dordrecht},
 {1994}.


\bibitem{Lichnz}
  {A. Lichnerowicz},
  \textit	 {L'int\'egration des \'equations de la gravitation relativiste et le probl\`eme des {$n$} corps},
 {J. Math. Pures Appl. (9)},
 {23},
 {1944},
 {37-63}.


\bibitem{Maxwell05}
    {D. Maxwell},
     \textit {Rough solutions of the {E}instein constraint equations on
              compact manifolds},
   {J. Hyperbolic Differ. Equ.},
  {2},
  {2005},
   {2},
   {521-546}.


\bibitem {Maxwell}
    {D. Maxwell},
     \textit {A class of solutions of the vacuum {E}instein constraint
              equations with freely specified mean curvature},
   {Math. Res. Lett.},
   {16},
   {2009},
   {4},
   {627-645}.


\bibitem{Rudin}
    {W. Rudin},
     \textit {Functional analysis},
   {International Series in Pure and Applied Mathematics},
   {Second},
  {McGraw-Hill Inc.},
  {New York},
  {1991},
 {xviii+424}.
 

\bibitem{TK}
    {I. Termwuttipong and T.Kaewtem},
     \textit {Fixed point theorem of half-continuous mappings on topological
              vector spaces},
   {Fixed Point Theory Appl.},
   {2010},
  {Art. ID 814970, 10}.
  
\end{thebibliography}
\end{document}